\documentclass[aps,prd,preprint,amsmath,amssymb,nofootinbib,floatfix]{revtex4-2}

\usepackage{graphicx}
\usepackage{amsthm}
\usepackage{hyperref}
\usepackage{booktabs}



\newtheorem{theorem}{Theorem}

\begin{document}

\title{A minimization theorem for the Koide ratio and its Standard Model
       calibration}

\author{K.~H\"ubner}
\thanks{Corresponding author: \href{mailto:k.a.huebner@web.de}{k.a.huebner@web.de}}
\thanks{ORCID: \href{https://orcid.org/0009-0001-6425-9527}{0009-0001-6425-9527}}
\affiliation{Independent researcher}

\date{\today}

\begin{abstract}
The charged-lepton Koide relation remains a striking empirical regularity in Standard-Model flavor data. We prove that for any positive mass set with Koide ratio $Q_0$, the one-particle extension $Q(m_1,\ldots,m_N,x)$ has a unique global minimum $Q_\text{min}=Q_0/(1+Q_0)$ at $m^*=\bigl[(\sum_i m_i)/(\sum_i \sqrt{m_i})\bigr]^2$. This exact kinematic result defines a unique extension benchmark. For the measured charged leptons it gives $m_*^\ell = 1.255\,34(16)\,\text{GeV}$ and $Q_{4,\min}^{\mathrm{exp}} = 0.399\,997\,8(43)$; in the ideal Koide limit $Q_\ell^{\mathrm{K}}=2/3$, the corresponding minimum is exactly $2/5$. In the effective-participant language $N_{\mathrm{eff}}\equiv 1/Q$, the optimal one-particle extension increases $N_{\mathrm{eff}}$ by one, while the equal-$k$ multiplet extension increases it by $k$. The one-particle $N_{\mathrm{eff}}$ profile is exactly Lorentzian in a dimensionless share-mismatch coordinate $u$, which we interpret kinematically rather than dynamically. Using charged-lepton pole masses with the PDG~2024 own-scale $\overline{\text{MS}}$ charm mass gives $Q(e,\mu,\tau,c)=0.400\,002\,5(64)$, i.e. $11.7\,\text{ppm}$ above the measured-input benchmark and $6.2\,\text{ppm}$ above $2/5$. This intentionally mixed-definition comparison is treated only as a phenomenological coincidence. To calibrate it within a stated benchmark class, we perform an exhaustive common-scale scan over non-neutrino Standard Model 2-body and 3-body seeds with one added mass. The charged-lepton-plus-charm continuation ranks $33/12{,}720$ in the raw trial set, $24/2{,}640$ after collapsing repeated scale realizations, and $6/756$ within the fermion-only collapsed subset. We present the charm case as an empirically calibrated example of the theorem, not as a dynamical flavor model.
\end{abstract}

\maketitle

\section{Introduction}
\label{sec:intro}

In 1982 Koide observed that the three charged lepton masses satisfy,
to remarkable accuracy, the dimensionless relation~\cite{Koide1982,Koide1983}
\begin{equation}
  Q \;\equiv\;
  \frac{m_e + m_\mu + m_\tau}
       {(\sqrt{m_e}+\sqrt{m_\mu}+\sqrt{m_\tau})^2}
  \;=\; \frac{2}{3}\,,
  \label{eq:koide}
\end{equation}
with a relative deviation of $9.3\,\text{ppm}$ using PDG~2024 lepton
masses~\cite{PDG2024}.
Because $Q$ is homogeneous of degree zero, it probes
the relative placement of masses rather than their overall scale.

No symmetry of the Standard Model predicts Eq.~\eqref{eq:koide}.
Foot reformulated the charged-lepton relation geometrically in terms of
the angle between the $\sqrt m$ vector and the democratic
direction~\cite{Foot1994}, while Sumino proposed a family-gauge
mechanism that would protect the pole-mass relation against QED
radiative corrections~\cite{Sumino2009a,Sumino2009b}, a proposal that
is intrinsically tied to additional family-gauge dynamics not without
phenomenological constraints. Various Koide-like
extensions to running masses, quarks, neutrinos, mixing-dependent
pseudo-masses, and generalized Koide functionals have also been
explored~\cite{LiMa2005,GerardGoffinetHerquet2006,LiMa2006,
XingZhang2006,Rodejohann2011,Kartavtsev2011,Gao2015,Rivero2005,
Brannen2010,Varma2026}, with mixed phenomenological success. Because any a priori probability estimate for
the charged-lepton coincidence depends strongly on the assumed prior over
positive mass tuples, we do not assign a formal global significance to
Eq.~\eqref{eq:koide} here; instead, the only statistical calibration in
this paper is the explicit look-elsewhere scan of
Sec.~\ref{sec:scan}.

This paper addresses a narrower question than a full flavor model:
given a fixed positive mass multiplet, which extra mass minimizes the
extended Koide ratio, and what does that imply for the observed charged
leptons? We show that the answer is exact and model-independent: for any
positive base set, $Q(m_1,\ldots,m_N,x)$ has a unique global minimum at
$x=m^*$, and the minimum value is $Q_0/(1+Q_0)$. Applied to the charged
leptons, the theorem yields the measured-input benchmark
$Q_{4,\min}^{\mathrm{exp}}$ at the scale $m_*^\ell$; in the idealized
world where the lepton Koide value is exactly $2/3$, the corresponding
four-body minimum is exactly $2/5$.
Accordingly, we reserve
$Q_\ell^{\mathrm{exp}} = Q(e,\mu,\tau)=0.666\,661(12)$ for the measured
charged-lepton ratio and $Q_\ell^{\mathrm{K}} = 2/3$ for the idealized
exact Koide value.

The theorem itself is the main new result of the paper. Earlier Koide
literature has discussed empirical extensions to quarks, neutrinos,
running masses, mixing-dependent pseudo-masses, and generalized
functionals~\cite{LiMa2005,
GerardGoffinetHerquet2006,LiMa2006,XingZhang2006,Rodejohann2011,
Kartavtsev2011,Gao2015,Rivero2005,Brannen2010,Varma2026}; the point here is more
specific. Rather than proposing another extension ansatz, we derive the
unique extension point selected by minimizing the Koide ratio once a base
tuple is held fixed.

Two reformulations of the same theorem turn out to be useful and are
also new in this Koide-specific context. First, an equal-amplitude
multiplet extension: when $k$ additional masses are added with a common
amplitude $r=\sqrt{x}$, the minimum is reached at the same
$r^*=R_2/R_1$ as in the one-particle case and takes the closed form
$Q_{\min}^{(k)}=Q_0/(1+kQ_0)$. Second, an
\emph{effective-participant} reading
$N_{\mathrm{eff}}\equiv 1/Q$ in which the one-particle theorem
collapses to the additive identity
$N_{\mathrm{eff},\max}=N_{\mathrm{eff},0}+1$ and its multiplet
generalization to $N_{\mathrm{eff},\max}^{(k)}=N_{\mathrm{eff},0}+k$;
in particular the idealized lepton seed $N_{\mathrm{eff},0}^{\mathrm{K}}=3/2$
maps to the four-body benchmark $N_{\mathrm{eff},\max}^{\mathrm{K}}=5/2$.
This reading exposes that the kinematic content of the Koide ratio is a
participation count rather than an energy scale.
A third reformulation reparameterizes the one-particle extension curve
in an exact Lorentzian coordinate
$u(x)=(p-Q_{N+1,\min})/\sqrt{Q_{N+1,\min}(1-Q_{N+1,\min})}$, where
$p=r/(R_1+r)$ is the normalized amplitude share of the extender. This is
an exact reparameterization of the kinematic curve, not a dynamical
resonance, and gives the cleanest way to quote how close any candidate
continuation sits to the theorem-selected peak. The same minimum is
equivalently encoded geometrically in the Foot angle~\cite{Foot1994}
between the $\sqrt m$ vector and the democratic direction.

The phenomenological question is then whether any Standard Model mass
lies close to that theorem-selected extension point. Using
charged-lepton pole masses together with the short-distance charm mass
$m_c(m_c)$, the PDG~2024 charm input lies $1.4\%$ above $m_*^\ell$ and
gives $Q(e,\mu,\tau,c)=0.400\,002\,5$, which is
$11.7\,\text{ppm}$ above $Q_{4,\min}^{\mathrm{exp}}$ and
$6.2\,\text{ppm}$ above the idealized value $2/5$.
Because that comparison mixes lepton pole masses with a short-distance
quark mass, and because other conventional charm-mass choices move the
result away by $10^3$--$10^4\,\text{ppm}$, the observation must be
treated as a mixed-definition phenomenological coincidence rather than a
scheme-independent relation.

The main empirical task of this paper is therefore to calibrate that
coincidence rather than to promote it to a new flavor law. We perform an
exhaustive scan over all non-neutrino SM
2-body and 3-body seeds with one extra particle, evaluating any
quark-containing trial in a common-scale $\overline{\text{MS}}$ setup.
This shows whether the charged-lepton-plus-charm continuation is merely
one entry in a broad pool of accidental near-minima, or whether it ranks
high in this explicit benchmark ensemble. Charm is highlighted not
because it is the numerically best Standard Model extension, but because
it extends the already distinguished charged-lepton Koide triplet.

That benchmark ensemble is itself structured and partly a posteriori: it
restricts attention to non-neutrino SM masses, to one-particle
extensions of 2-body and 3-body seeds, and to the conventional common
scales $\mu\in\{2\,\text{GeV},m_c,m_b,m_Z,m_t\}$ for quark-containing
trials. The resulting scan fractions are therefore best read as rank
statistics within this chosen comparison class, not as prior
probabilities for flavor models or spectra.

The paper is organized as follows. Section~\ref{sec:theorem} proves the
exact one-particle minimization theorem and its equal-amplitude
multiplet extension. Section~\ref{sec:neff_lorentzian} develops the
concentration-measure, effective-participant ($N_{\mathrm{eff}}=1/Q$),
and exact Lorentzian-coordinate reformulations of the extension problem,
including the idealized
$Q_\ell^{\mathrm{K}}=2/3 \to Q_{4,\min}^{\mathrm{K}}=2/5$
(equivalently $N_{\mathrm{eff},0}^{\mathrm{K}}=3/2
\to N_{\mathrm{eff},\max}^{\mathrm{K}}=5/2$) corollary.
Section~\ref{sec:charm} then applies that setup to the measured charged
leptons and quantifies the charm comparison.
Section~\ref{sec:scheme} examines the scheme dependence of that
mixed-definition input choice. Section~\ref{sec:foot-angle} recasts the
charm discussion in Foot's angular language and quantifies the
theorem-selected continuation angle. Section~\ref{sec:scan} gives
the look-elsewhere calibration from the exhaustive 2-body and 3-body scan,
where the charged-lepton-plus-charm continuation ranks $24/2{,}640$ in
the full collapsed leaderboard and $6/756$ within the fermion-only
sub-ensemble.
Section~\ref{sec:conclusions} summarizes the outcome and its limitations.

We use natural units throughout.
Quark masses are $\overline{\text{MS}}$ values from
PDG~2024~\cite{PDG2024} unless otherwise stated.
Lepton masses are on-shell (pole) masses, which is also the empirical
input convention in which the charged-lepton Koide relation is usually
quoted. When charm is added, we use the short-distance quantity
$m_c(m_c)$ as the default reference because a charm pole mass is not a
comparably clean observable. The resulting lepton-pole plus
$\overline{\text{MS}}$-charm comparison is therefore a practical
mixed-definition convention rather than a model-derived common scheme.
Unless stated otherwise, quoted ppm deviations from a reference value
$X_\text{ref}$ mean the signed relative difference
$10^6(X-X_\text{ref})/X_\text{ref}$; when only closeness matters, we
quote the magnitude. Phrases such as ``$X$ ppm above $X_\text{ref}$''
are used as a verbal shorthand for a positive value of the same signed
relative difference, not as a separate inequality statement.
Quoted uncertainties are propagated at leading order from the stated input
mass errors using diagonal Gaussian error propagation. We neglect input
correlations throughout and use the resulting numbers only as local
sensitivity estimates rather than as full statistical intervals.
Neutrinos are excluded from the scan below because only mass splittings,
not a unique absolute-mass spectrum, are experimentally established; any
neutrino leaderboard would therefore build in additional assumptions that
are not needed for the present charged-lepton benchmark.

\section{The \texorpdfstring{$Q$}{Q}-minimum theorem}
\label{sec:theorem}

\begin{theorem}
  Let $m_i>0$ be particle masses and let $r_i = \sqrt{m_i}$ for
  $i=1,\ldots,N$. Define
  \[
    R_1 \equiv \sum_{i=1}^N r_i\,,
    \qquad
    R_2 \equiv \sum_{i=1}^N r_i^2 = \sum_{i=1}^N m_i\,,
    \qquad
    Q_0 \equiv Q_N(\{m_i\}) = \frac{R_2}{R_1^2}\,.
  \]
  Here $Q_0$ is simply the Koide ratio of the unextended base tuple.
  For an added mass $x$ written as $r=\sqrt{x}>0$, consider the extended ratio
  \begin{equation}
    Q_{N+1}(r) = \frac{R_2 + r^2}{(R_1+r)^2}\,.
    \label{eq:q_extension_r}
  \end{equation}
  Then:
  \begin{enumerate}
    \item[(i)] $Q_{N+1}$ has a unique global minimum at
      \begin{equation}
        r^* \;=\; \frac{R_2}{R_1}
        \;=\; Q_0 R_1\,,
        \qquad
        m^* \;=\; (r^*)^2
        \;=\; \left(\frac{\sum_i m_i}{\sum_i \sqrt{m_i}}\right)^{\!2}\,,
        \label{eq:mstar}
      \end{equation}
    \item[(ii)] The minimum value is
      \begin{equation}
        Q_\text{min}
        \;=\; \frac{Q_0}{1 + Q_0}\,.
        \label{eq:qmin}
      \end{equation}
  \end{enumerate}
\end{theorem}

\begin{proof}
  Differentiating Eq.~\eqref{eq:q_extension_r} with respect to $r$ gives
  \begin{equation}
    \frac{\mathrm{d} Q_{N+1}}{\mathrm{d}r}
    = \frac{2(R_1r-R_2)}{(R_1+r)^3}
    = \frac{2R_1(r-r^*)}{(R_1+r)^3}\,.
    \label{eq:dqdr}
  \end{equation}
  In particular, $\mathrm{d}Q_{N+1}/\mathrm{d}r = 0$ at $r=r^*$.
  Since $R_1>0$ and $R_1+r>0$, the derivative is negative for $0<r<r^*$
  and positive for $r>r^*$, so $r^*$ is the unique global minimum.
  Substituting $R_2 = Q_0R_1^2$ and $r^*=Q_0R_1$ gives
  \begin{align}
    Q_\text{min}
    &= \frac{R_2 + (r^*)^2}{(R_1+r^*)^2}
     = \frac{Q_0R_1^2 + Q_0^2R_1^2}{R_1^2(1+Q_0)^2} \nonumber\\
     &= \frac{Q_0 + Q_0^2}{(1+Q_0)^2}
      = \frac{Q_0}{1+Q_0}\,. \qed
  \end{align}
\end{proof}

For any positive $(N+1)$-tuple, the Cauchy--Schwarz inequality gives
$R_1^2 \le (N+1)R_2$ and hence $Q_{N+1} \ge 1/(N+1)$.
For any positive $(N+1)$-tuple with $N+1\ge 2$, one has
$R_1^2 = R_2 + 2\sum_{i<j} r_i r_j > R_2$, so $Q_{N+1}<1$; for the
trivial one-mass case one instead has $Q_1=1$.
The theorem above does not require the $m_i$ to be distinct; degenerate
base masses ($m_i=m_j$ for some $i\ne j$) are allowed, in which case the
$r_i$-weighted barycenter $r^*$ remains a well-defined positive number.
With the first $N$ masses fixed, however, an extension generally cannot
reach the absolute lower bound $1/(N+1)$; its sharp lower bound is instead
the theorem value $Q_\text{min}=Q_0/(1+Q_0)$. Since every positive
$N$-tuple satisfies
$Q_0\ge 1/N$, one has
\begin{equation}
  Q_\text{min} = \frac{Q_0}{1+Q_0} \ge \frac{1}{N+1}\,.
\end{equation}
Equality holds only in the trivial equal-mass limit.

Three brief readings of $r^*=R_2/R_1$ are useful. First,
$R_1r^*-R_2=\sum_i r_i(r^*-r_i)=0$ shows that $r^*$ is the
$r_i$-weighted barycenter of the base $r_i$ on the positive line.
Equivalently, if one regards the $r_i$ as a weighted sample with weights
proportional to $r_i$, then $r^*$ is their weighted mean and the minimum
condition is the vanishing of the corresponding weighted first central moment.
Because $r^*$ is a positive weighted mean of the base mass amplitudes
$r_i=\sqrt{m_i}$, it necessarily lies between the smallest and largest
$r_i$; equivalently, the minimizing mass $m^*=(r^*)^2$ lies between the
smallest and largest base masses.
Second, in Foot's democratic-vector language~\cite{Foot1994}, the allowed extension line
$W(r)=(r_1,\ldots,r_N,r)$ reaches its least-deviating alignment with the
democratic direction precisely at $r=r^*$. We return to this geometric
interpretation in Sec.~\ref{sec:foot-angle}.
Third, the theorem also admits an exact concentration-measure reading.
Because $Q_0=\sum_i p_i^2=\mathrm{tr}(\rho^2)$ for the normalized
amplitude weights $p_i=r_i/R_1$, it can be read as a purity or inverse
participation measure on the probability simplex.
We return to that interpretation, and to the associated effective-participant
coordinate $N_{\mathrm{eff}}\equiv 1/Q$, in Sec.~\ref{sec:neff_lorentzian}.

The same minimizer also controls the equal-multiplet extension problem:
if one adds $k$ new masses and minimizes over their corresponding
$r$-variables jointly, the global minimum occurs when they are all equal
to the same $r^*$, giving $Q_{\min}^{(k)} = Q_0/(1+kQ_0)$.
To see this, let the new mass amplitudes be $s_1,\ldots,s_k>0$ and define
$S\equiv \sum_{a=1}^k s_a$. For fixed $S$, one has
$\sum_a s_a^2 \ge S^2/k$, with equality only when all $s_a=S/k$.
The multi-particle problem therefore reduces to minimizing
\begin{equation}
  \widetilde Q_{N+k}(S) = \frac{R_2 + S^2/k}{(R_1+S)^2}
\end{equation}
over $S>0$. Differentiating gives
\begin{equation}
  \frac{\mathrm{d}\widetilde Q_{N+k}}{\mathrm{d}S}
  = \frac{2(R_1S-kR_2)}{k(R_1+S)^3}\,,
\end{equation}
so the unique minimum occurs at $S^*=kR_2/R_1$, i.e. at
$s_1=\cdots=s_k=R_2/R_1=r^*$.
Substituting back yields
\begin{equation}
  Q_{\min}^{(k)}
  = \frac{R_2 + k(r^*)^2}{(R_1+kr^*)^2}
  = \frac{Q_0}{1+kQ_0}\,.
\end{equation}
Thus repeated exact minimization keeps
returning to the same optimal scale and approaches the large-$N$
lower bound $Q_N \sim 1/N$ from above.

The same statement also applies to generalized Koide functionals of the
type considered in Ref.~\cite{Varma2026}. For
\begin{equation}
  R_{k,q}(m_i)
  \equiv
  \frac{\bigl(\sum_i m_i^{1/k}\bigr)^{kq}}
       {\bigl(\sum_i m_i^{1/(2k)}\bigr)^{2kq}}\,,
\end{equation}
set $z_i=m_i^{1/k}$, so that
$R_{k,q}=Q(z_i)^{kq}$. For $kq>0$, minimizing $R_{k,q}$ over one added
positive mass is therefore equivalent, by monotonicity, to minimizing the
ordinary Koide ratio in the $z$ variables. The selected added mass in the
original variables is
\begin{equation}
  x^*_{k,q}
  = \left(\frac{\sum_i m_i^{1/k}}{\sum_i m_i^{1/(2k)}}\right)^{2k},
\end{equation}
and the minimum is $[Q_0/(1+Q_0)]^{kq}$, where
$Q_0=Q(m_i^{1/k})$. The equal-multiplet extension similarly gives
$[Q_0/(1+\ell Q_0)]^{kq}$ for $\ell$ equal added amplitudes in the
$z$-space problem. This is only a change-of-variables corollary of the
theorem, not an additional dynamical input.

Equation~\eqref{eq:mstar} is the only ingredient needed for the
phenomenology below: once a base tuple is fixed, the theorem selects a
unique extension scale $m^*$ and an equally unique benchmark value
$Q_\text{min}$. The question is then empirical rather than formal:
whether any known SM mass lies unusually close to that theorem-selected
point.
If the charged-lepton Koide relation reflects real flavor structure
rather than accident, then this minimizer can be read as the unique
least-disruptive one-particle continuation of that structure. We do not
develop such a model here; we use the theorem only as a benchmark.

\section{Effective-Participant and Lorentzian Coordinates}
\label{sec:neff_lorentzian}

Beyond the minimizing mass itself, the theorem has useful exact corollaries
for concentration measures and for the idealized charged-lepton Koide limit.
In this section $Q_\ell^{\mathrm{K}}=2/3$ denotes the idealized exact
charged-lepton Koide value; the corresponding measured-input benchmarks are
introduced later in Sec.~\ref{sec:charm}.

The theorem has a compact reading in terms of concentration measures.
Define normalized amplitude weights $p_i = r_i/R_1 \ge 0$ with
$\sum_i p_i = 1$; the weights live on the probability simplex
$\Delta^{N-1}$.
Packaging them into the diagonal probability matrix
$\rho \equiv \mathrm{diag}(p_1,\ldots,p_N)$, one has
\begin{equation}
  Q_0 = \mathrm{tr}(\rho^2)
      = \sum_i p_i^2\, ,
\end{equation}
equivalently the Herfindahl--Hirschman concentration index~\cite{DOJFTC2010}
or the squared $\ell^2$ norm of $\{p_i\}$.
In the present classical setting this same quantity is simply the collision
probability, i.e. the probability that two independent draws from
$\{p_i\}$ return the same family index.
It is smallest at the equal-mass point $p_i=1/N$, where
$\rho=\mathbf{1}/N$ and $Q_0=1/N$, and it tends to $1$ as a single mass
dominates.
The exact charged-lepton Koide value $Q_\ell^{\mathrm{K}}=2/3$ is therefore
the statement that the charged-lepton $\sqrt m$ profile has purity $2/3$,
well above the three-mass equal-mass floor of $1/3$.

The extended ratio is
\begin{equation}
  Q_{N+1}(r) = \frac{R_2+r^2}{(R_1+r)^2}\, ,
\end{equation}
minimized at $r^* = Q_0R_1 = R_2/R_1$.
So the minimizing amplitude is the base purity times the total base
amplitude $R_1$.
The resulting amplitude profile is
\begin{equation}
  p_i^* = (1-Q_\text{min})p_i,
  \qquad
  p_{N+1}^* = Q_\text{min}
\end{equation}
so the old distribution is uniformly rescaled while the new particle absorbs
exactly the final purity. The rescaling factor follows from
$R_1+r^* = R_1(1+Q_0) = R_1/(1-Q_{\min})$.

Because $Q$ is itself a purity, its inverse has a direct counting
interpretation, and the theorem takes its simplest form in that
language. Define the \emph{effective participant number}
\begin{equation}
  N_{\mathrm{eff}} \equiv \frac{1}{Q}\,.
\end{equation}
It equals $N$ for the uniform distribution (all $N$ amplitudes contribute
equally) and approaches $1$ as one amplitude dominates.
Equivalently, $N_{\mathrm{eff}}$ is the size of the hypothetical equal-weight
distribution that has the same purity as $\{p_i\}$.
For the base tuple $N_{\mathrm{eff},0}=1/Q_0$ and for the peak value reached
at the minimizing extension $N_{\mathrm{eff},\max}=1/Q_\text{min}$, the relation
$Q_\text{min}=Q_0/(1+Q_0)$ is simply
\begin{equation}
  N_{\mathrm{eff},\max} = N_{\mathrm{eff},0} + 1\,.
  \label{eq:inverse_q_identity}
\end{equation}
So the optimal one-particle extension increases the effective participant
number by exactly one.
For the exact charged-lepton Koide value $Q_\ell^{\mathrm{K}}=2/3$, one has
$N_{\mathrm{eff},0}^{\mathrm{K}}=3/2$; after the optimal extension,
\begin{equation}
  Q_{4,\min}^{\mathrm{K}} = \frac{Q_\ell^{\mathrm{K}}}{1+Q_\ell^{\mathrm{K}}}
  = \frac{2}{5}
\end{equation}
and therefore $N_{\mathrm{eff},\max}^{\mathrm{K}}=5/2$.

The same language extends immediately to the equal-$k$ multiplet theorem.
If $k$ new amplitudes are added and minimized jointly, the result
$Q_{\min}^{(k)}=Q_0/(1+kQ_0)$ reads
\begin{equation}
  N_{\mathrm{eff},\max}^{(k)} = N_{\mathrm{eff},0} + k\,.
  \label{eq:inverse_q_identity_k}
\end{equation}
Each of the $k$ optimal particles therefore contributes exactly one unit to
the effective participant number.

The same effective-participant language also admits a useful exact shape
reparameterization for the one-particle extension curve.
For fixed base amplitudes and an added amplitude $r=\sqrt{x}$,
\begin{equation}
  N_{\mathrm{eff}}(r) = \frac{(R_1+r)^2}{R_2+r^2}\,.
\end{equation}
Defining the dimensionless extension strength
\begin{equation}
  s \equiv \frac{r}{R_1}
\end{equation}
and the new particle's normalized amplitude share
\begin{equation}
  p_X \equiv \frac{r}{R_1+r} = \frac{s}{1+s}\, ,
\end{equation}
one may introduce the standardized share-mismatch coordinate
\begin{equation}
  u \equiv \frac{p_X-Q_{\min}}{\sqrt{Q_{\min}(1-Q_{\min})}}\,.
  \label{eq:u_coordinate}
\end{equation}
Here $s$, $p_X$, and $u$ are all dimensionless.
The variable $s$ measures the new amplitude against the total coherent
base amplitude, $p_X$ is the new state's share of the total normalized
$\sqrt m$ weight, and $u$ measures the deviation of that share from the
theorem-selected value in units of the natural Bernoulli width
$\sqrt{Q_{\min}(1-Q_{\min})}$.
A direct calculation gives the exact identity
\begin{equation}
  Q(p_X) = p_X^2 + Q_0(1-p_X)^2
        = Q_{\min} + \frac{(p_X-Q_{\min})^2}{1-Q_{\min}}\,,
  \label{eq:q_quadratic_in_share}
\end{equation}
i.e. $Q$ is exactly quadratic in the share $p_X$, with minimum
$Q_{\min}=Q_0/(1+Q_0)$ at $p_X=Q_{\min}$. Substituting the $u$
coordinate of Eq.~\eqref{eq:u_coordinate} into
Eq.~\eqref{eq:q_quadratic_in_share} gives
$Q(u)=Q_{\min}(1+u^2)$, and hence
\begin{equation}
  N_{\mathrm{eff}}(u) = \frac{N_{\mathrm{eff},\max}}{1+u^2}
  \qquad\text{with}\qquad
  N_{\mathrm{eff},\max} = \frac{1}{Q_{\min}}\,.
  \label{eq:neff_lorentzian}
\end{equation}
The Lorentzian form is therefore not an independent observation but a
direct algebraic consequence of the exact quadratic dependence
\eqref{eq:q_quadratic_in_share} together with the inversion
$N_{\mathrm{eff}}=1/Q$.
The normalized peak is an exact Lorentzian,
$N_{\mathrm{eff}}/N_{\mathrm{eff},\max}=1/(1+u^2)$.
Its maximum occurs at $u=0$, where
$N_{\mathrm{eff}}(0)=N_{\mathrm{eff},\max}$.
Its half-maximum occurs at $|u|=1$, and the inflection points are the
universal values $u=\pm 1/\sqrt{3}$, independent of the base tuple.
This universal profile will be evaluated for the measured charged-lepton
base in Sec.~\ref{sec:charm}.

This Lorentzian form should be read kinematically, not dynamically.
It does \emph{not} imply an unstable state, a decay width, or a literal
Breit--Wigner resonance in flavor space.
Its value is instead interpretive: after the nonlinear coordinate change
of Eq.~\eqref{eq:u_coordinate}, the entire one-particle extension problem
is reduced to a universal even function of the deviation from the
theorem-selected amplitude share.
In that sense $u$ is a natural dimensionless mismatch coordinate for any
future dynamical model that attempts to realize the least-disruptive
continuation of a Koide-special base tuple.

\section{Lepton extension and the charm scale}
\label{sec:charm}

We now turn to the measured charged-lepton spectrum.
In this section, $Q_\ell^{\mathrm{exp}}$ denotes the charged-lepton
Koide ratio computed from the measured PDG~2024 pole masses, while
$Q_\ell^{\mathrm{K}}=2/3$ remains the idealized exact Koide value used
for comparison, with corresponding theorem-selected four-body benchmark
$Q_{4,\min}^{\mathrm{K}}=Q_\ell^{\mathrm{K}}/(1+Q_\ell^{\mathrm{K}})=2/5$.
Using the PDG~2024 charged-lepton pole masses and hence
$Q_\ell^{\mathrm{exp}}=0.666\,661(12)$~\cite{PDG2024},
\begin{equation}
  Q_{4,\min}^{\mathrm{exp}}
  = \frac{Q_\ell^{\mathrm{exp}}}{1+Q_\ell^{\mathrm{exp}}}
  = 0.399\,997\,8(43)\,.
  \label{eq:qmin_leptons_num}
\end{equation}
Equivalently, the measured charged-lepton seed has
$N_{\mathrm{eff},0}^{\mathrm{exp}}=1/Q_\ell^{\mathrm{exp}}=1.500014(27)$ and the
theorem-selected four-body peak value is
$N_{\mathrm{eff},\max}^{\mathrm{exp}}=1/Q_{4,\min}^{\mathrm{exp}}=2.500014(27)$.
The measured spectrum also fixes the corresponding extension scale through
Eq.~\eqref{eq:mstar}. For the charged leptons,
\begin{equation}
  r_*^\ell = \frac{R_2}{R_1} = 1.120\,42(7)\,\text{GeV}^{1/2}
\end{equation}
and therefore
\begin{equation}
  m_*^\ell = (r_*^\ell)^2
  = \left(\frac{m_e + m_\mu + m_\tau}
           {\sqrt{m_e}+\sqrt{m_\mu}+\sqrt{m_\tau}}\right)^{\!2}
  = 1.255\,34(16)\,\text{GeV}\,,
  \label{eq:mstar_num}
\end{equation}
using PDG~2024 lepton masses~\cite{PDG2024}.
If one substitutes the idealized value $Q_\ell^{\mathrm{K}}=2/3$ for
$Q_\ell^{\mathrm{exp}}$ in the closed-form identity
$m^*=Q_0(m_e+m_\mu+m_\tau)$, while keeping the same measured mass sum
$R_2=m_e+m_\mu+m_\tau$, one obtains the corresponding idealized
extension scale
\begin{equation}
  m_*^{\ell,\mathrm{K}} = Q_\ell^{\mathrm{K}}(m_e+m_\mu+m_\tau)
  = \frac{2}{3}(m_e+m_\mu+m_\tau)
  = 1.255\,35(14)\,\text{GeV}\,.
\end{equation}
Here $m^*=(R_2/R_1)^2 = Q_0R_2$ is the same identity as in the main
text; the only knob being turned is $Q_0$, no independent assumption
on the lepton spectrum is added.
This is numerically almost identical to $m_*^\ell$, the small residual
difference simply reflecting $Q_\ell^{\mathrm{K}}-Q_\ell^{\mathrm{exp}}$.
The PDG~2024 reference $\overline{\text{MS}}$ charm mass
$m_c(m_c) = 1.273(9)\,\text{GeV}$ is $1.41\%$ above $m_*^\ell$.

The full dependence of $Q(e,\mu,\tau,x)$ on the fourth mass is shown
in Fig.~\ref{fig:lepton_extension}. The curve exhibits the unique
minimum at $x = m_*^\ell$, the stronger lepton-data-specific lower
bound $Q \ge Q_{4,\min}^{\mathrm{exp}}$, and the proximity of the charm
mass to that minimum. No other Standard Model particle mass falls within
several percent of $m_*^\ell$: the nearest are
$m_\tau = 1.777\,\text{GeV}$ (41\% above) and the $b$ quark
($\times 3.3$ too heavy).

\begin{figure}[!t]
  \centering
  \includegraphics[width=0.94\linewidth]{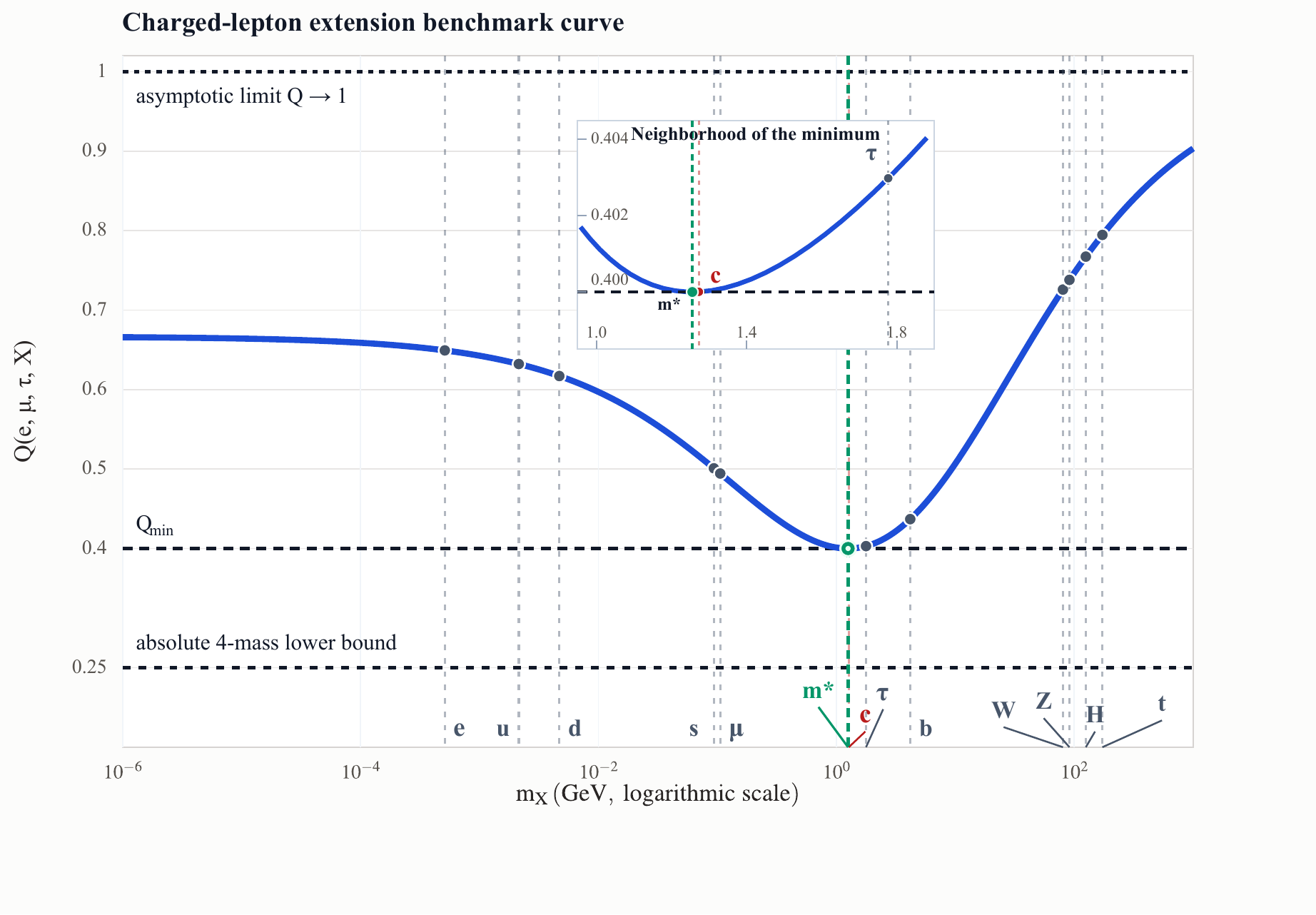}
  \caption{Extended Koide ratio $Q(e,\mu,\tau,x)$ as a function of a
  hypothetical fourth mass $x$. The unique minimum occurs at the
  theorem-selected mass $m_*^\ell = 1.255\,34(16)\,\text{GeV}$, while the
  physical charm mass lies nearby and gives $Q(e,\mu,\tau,c)$ close to
  the measured-input minimum $Q_{4,\min}^{\mathrm{exp}}$; it is also close
  to $2/5$, which is the idealized target obtained in the limit
  $Q_\ell=2/3$.
  The horizontal lines show the generic four-body lower bound $1/4$,
  the measured-input minimum $Q_{4,\min}^{\mathrm{exp}}$, the idealized
  four-body minimum $2/5$, and the asymptotic upper limit $Q \to 1$.
  Standard Model mass markers are shown for orientation. The inset resolves the
  charm-scale neighborhood of the minimum. In the Lorentzian coordinate
  of Sec.~\ref{sec:neff_lorentzian}, the charm point sits at
  $u_c=+0.0034(17)$, i.e. essentially at the peak center.}
  \label{fig:lepton_extension}
\end{figure}

Because $Q_{4}(x)$ is stationary at $x=m_*^\ell$, its leading variation is
quadratic in $(\sqrt{x}-\sqrt{m_*^\ell})$ near the minimum. Writing
$r=\sqrt{x}=r_*^\ell+\epsilon$ and using
$Q''(r_*^\ell)=2R_1/(R_1+r_*^\ell)^3$ from Eq.~\eqref{eq:dqdr}, one finds
\begin{equation}
  Q_4(r) = Q_{4,\min}^{\mathrm{exp}} +
  \frac{R_1\,\epsilon^2}{(R_1+r_*^\ell)^3} + O(\epsilon^3)\,.
\end{equation}

Thus a fractional offset
$\delta = (m_c - m_*^\ell)/m_*^\ell = 0.0141$ in the mass contributes only
a few parts in $10^6$ to $Q$.
Direct evaluation gives
\begin{equation}
  Q(e,\mu,\tau,c) \;=\; 0.400\,002\,5(64)
  \label{eq:Qext}
\end{equation}
using $m_c(m_c) = 1.273(9)\,\text{GeV}$.
This corresponds to a relative deviation of $6.2\,\text{ppm}$ from
$Q_{4,\min}^{\mathrm{K}} = 2/5$, and an $11.7\,\text{ppm}$ relative
deviation from the measured-input minimum
$Q_{4,\min}^{\mathrm{exp}}$. The latter is the direct benchmark selected
by the theorem for the observed lepton inputs, while the former is the
corresponding ideal-Koide limit. In the Lorentzian coordinate of
Sec.~\ref{sec:neff_lorentzian}, the same point corresponds to
$u_c=+0.0034(17)$, i.e. essentially the center of the exact profile.
In this sense, the charm mass provides a numerically close Standard Model
realization of the theorem-selected minimizing point, although the theorem
itself does not identify a physical SM state.
Here the mass prescription is again intentionally pragmatic rather than
derived: the charged leptons are kept at their pole masses because the
observed Koide relation is formulated for those inputs, while charm is
taken as the standard short-distance quantity $m_c(m_c)$ because the
charm pole mass is less clean conceptually and numerically.

The observed extenders map to the Lorentzian coordinate values shown in
Table~\ref{tab:ucoords}. Charm lies essentially at the peak center,
whereas the electroweak and top masses all lie near the $u\sim 1$ shoulder
where the Lorentzian has already fallen to about half height.

\begin{table}[t]
  \centering
  \caption{Lorentzian-coordinate values $u_X$ for one-particle extensions of
  the charged-lepton base $(e,\mu,\tau)$, using the definition in
  Eq.~\eqref{eq:u_coordinate}. The quoted uncertainties are propagated from
  the input mass uncertainties.}
  \label{tab:ucoords}
  \begin{tabular}{lc}
    \toprule
    Extender $X$ & $u_X$ \\
    \midrule
    $u$ & $-0.7616(56)$ \\
    $d$ & $-0.7367(25)$ \\
    $s$ & $-0.5024(85)$ \\
    $c$ & $+0.0034(17)$ \\
    $b$ & $+0.3040(9)$ \\
    $W$ & $+0.90251(3)$ \\
    $Z$ & $+0.91926(2)$ \\
    $H$ & $+0.95824(11)$ \\
    $t$ & $+0.99322(62)$ \\
    \bottomrule
  \end{tabular}
\end{table}

Within the Standard Model spectrum, the charm quark is the only particle
lying comparably close to this minimizing scale, even though it is not the
best-ranked entry in the broader scan of Sec.~\ref{sec:scan}.
A useful sanity check is the one-body problem $Q(m,X)$: because
$Q(m)=1$ for any single positive mass, the theorem reduces to the trivial
statement that the minimizing extension recycles the original mass and
gives $Q_{\min}=1/2$. This serves only as a control showing that
scheme-expanded proximity is cheap to obtain in the trivial case; what
makes the higher-body examples interesting is that they select nontrivial
new scales rather than merely recycling the original one.

\section{Renormalization-scheme sensitivity}
\label{sec:scheme}

The extended formula \eqref{eq:Qext} depends critically on the
renormalization scheme used for $m_c$.
This scheme sensitivity can be turned into a diagnostic.
Table~\ref{tab:scheme} shows $Q(e,\mu,\tau,c)$ for four natural
choices of $m_c$.

The $\overline{\text{MS}}$ mass evaluated at its own mass scale is the
conventional choice in Table~\ref{tab:scheme} that lies closest to the
measured-input minimum $Q_{4,\min}^{\mathrm{exp}}$.
The pole mass, which deviates by $5{,}013\,\text{ppm}$ from
$Q_{4,\min}^{\mathrm{exp}}$, is disfavoured
compared with the own-scale $\overline{\text{MS}}$ choice.
This is best read as a scheme-sensitivity diagnostic, not as a derivation
of a unique physical prescription: within this limited comparison, the
closest conventional realization happens to occur for the own-scale
$\overline{\text{MS}}$ charm mass.

The lepton masses, which enter as pole masses (the unique
scheme-independent definition for stable particles), are unambiguous.
For charm, by contrast, a short-distance mass is more appropriate than a
pole mass, and the use of $m_c(m_c)$ is consistent with standard
heavy-quark phenomenology~\cite{PDG2024}. This gives a limited physical
rationale for the hybrid input set used here, but it does not elevate
the result to a common-scheme or field-theoretic mass relation. A more
uniform comparison, for example using running fermion masses or Yukawa
couplings at a common high scale, would be a different question from the
one addressed here; in particular, the known running-mass versions of the
charged-lepton relation are already less sharp than the pole-mass
formula~\cite{XingZhang2006}. We therefore do not reinterpret the
present comparison as a common-scale running-mass test.
While no renormalization scale makes a quark mass directly analogous to
the charged-lepton pole masses, the standard short-distance quantity
$m_c(m_c)$ is the least arbitrary low-energy benchmark for charm,
because it is defined at the charm threshold and minimizes large
perturbative logarithms; we therefore use it only as a pragmatic
comparison convention, not as a common-scheme statement.
The explicit $M_Z$ exercise below is an external literature-based
cross-check using the running masses quoted in
Ref.~\cite{XingZhangZhou2008}; it is not part of the scan pipeline of
Sec.~\ref{sec:scan}.
The discrete scale set
$\mu\in\{2\,\text{GeV},m_c,m_b,m_Z,m_t\}$ is chosen simply to span the
standard low-scale reference point, the heavy-quark thresholds, and one
electroweak benchmark without introducing a denser a posteriori scale
scan. The qualitative conclusion is stable across this set: the own-scale
short-distance charm mass is the closest low-energy conventional choice,
while the proximity does not survive as a common-scale statement.

\begin{table}[!t]
  \centering
  \caption{Extended Koide ratio $Q(e,\mu,\tau,c)$ as a function of
  the charm mass scheme. Lepton pole masses are fixed throughout. The
  fourth column gives the corresponding Lorentzian mismatch coordinate
  $u_c$ from Eq.~\eqref{eq:u_coordinate} with propagated uncertainty.
  The $m_c$ and $Q$ entries carry $1\sigma$ uncertainties; the ppm column
  uses central values only and is included purely as a ranking diagnostic.
  Note in particular that for the $\mu=m_Z$ row the propagated $1\sigma$
  band on the ppm entry is itself of order $10^4\,\text{ppm}$ (driven by
  the $\delta m_c(M_Z)\approx 84\,\text{MeV}$ input uncertainty), so its
  ordering relative to the $2\,\text{GeV}$ and pole-mass rows is not
  sharply established by the central values alone.}
  \label{tab:scheme}
  \begin{tabular}{lccccc}
    \toprule
    Scheme & $m_c$ (GeV) & \multicolumn{1}{c}{$(m_c-m_*^\ell)/m_*^\ell$} & \multicolumn{1}{c}{$u_c$} & \multicolumn{1}{c}{$Q$} &
             \multicolumn{1}{c}{relative ppm from $Q_{4,\min}^{\mathrm{exp}}$} \\
    \midrule
    $\overline{\text{MS}}$ at $\mu = m_c$
      & $1.273(9)$ & $+1.41\%$ & $+0.0034(17)$ & $0.400\,002\,5(64)$ & $+12$ \\
    $\overline{\text{MS}}$ at $\mu = 2\,\text{GeV}$
      & $1.094(8)$ & $-12.86\%$ & $-0.0335(17)$ & $0.400\,445(46)$ & $+1{,}119$ \\
    $\overline{\text{MS}}$ at $\mu = m_Z$
      & $0.619(84)$ & $-50.69\%$ & $-0.166(30)$ & $0.4110(40)$ & $+27{,}429$ \\
    Pole mass
      & $1.67(5)$ & $+33.03\%$ & $+0.0708(75)$ & $0.40200(43)$ & $+5{,}013$ \\
    \bottomrule
  \end{tabular}
\end{table}

The same ordering is visible in the Lorentzian coordinate: the own-scale
$\overline{\text{MS}}$ choice has $|u_c|\ll 1$, while the alternative charm
prescriptions move to visibly larger mismatches on the same exact kinematic
profile. For the $m_Z$ row, the much larger $u_c$ error simply tracks the
quoted uncertainty on the external running-mass benchmark
$m_c(M_Z)=0.619(84)\,\text{GeV}$ used below.

If the one-particle minimizer is applied once more to the measured
lepton-plus-charm system, the preferred fifth mass returns to the same
neighborhood: one finds $m_{*,5}^{\ell+c} = 1.2624(36)\,\text{GeV}$
(propagated from $\delta m_\tau$ and $\delta m_c$),
only $0.86\%$ below the reference $m_c(m_c) = 1.273(9)\,\text{GeV}$, i.e.
by $m_c-m_{*,5}^{\ell+c} = 10.6 \pm 9.7\,\text{MeV}$ when the
uncertainties on both inputs are combined in quadrature. The two
quantities are therefore statistically compatible at about the
$1.1\sigma$ level rather than sharply distinguishable.
In that limited sense the recursive continuation returns to the same
charm-scale neighborhood rather than selecting a clearly new Standard
Model mass scale.

For orientation, one may nonetheless inspect the common-scale running
version at $\mu=M_Z$. Using the published running masses of quarks and
charged leptons in Ref.~\cite{XingZhangZhou2008}, one finds
\begin{equation}
  Q_\ell(M_Z)=Q(e,\mu,\tau;M_Z)=0.667\,928\,7(11)
\end{equation}
and
\begin{equation}
  m_*^\ell(M_Z)=1.235\,30(15)\,\text{GeV} \ ,
  \qquad
  Q_{4,\min}(M_Z)=\frac{Q_\ell(M_Z)}{1+Q_\ell(M_Z)}=0.400\,454\,0(40) \ .
\end{equation}
The corresponding charm extension gives
\begin{equation}
  Q(e,\mu,\tau,c;M_Z)=0.41098(39) \ ,
\end{equation}
which lies about $2.63(98)\times 10^4$ ppm above the theorem-selected
four-body minimum. Equivalently, $m_c(M_Z)=0.619(84)\,\text{GeV}$ lies about
$49.9(6.8)\%$ below $m_*^\ell(M_Z)$. This common-scale benchmark therefore
confirms, in line with the running-mass discussion of
Ref.~\cite{XingZhang2006}, that the sharp low-energy coincidence does not
survive as a comparable $M_Z$ statement.

The size of this $M_Z$ shift can be understood from the QED running of
the lepton masses themselves. At one-loop QED, all three charged leptons
share the same mass anomalous dimension $\gamma_m^{(\text{QED})} = (3/2\pi)\,\alpha_{\text{em}}(\mu)$,
because they carry the same electric charge. Their masses therefore
run by a common multiplicative factor,
$m_i(\mu) = R(\mu)\,m_i(\mu_0)$ for $i\in\{e,\mu,\tau\}$ with
$R(\mu) = \exp[-\!\int_{\mu_0}^{\mu} \gamma_m^{(\text{QED})}(\mu')\,d\ln\mu']$.
Because $R$ is flavor-universal, $R_2(\mu) = R(\mu)^2 R_2$ and
$R_1(\mu) = R(\mu) R_1$, so
\begin{equation}
  Q_\ell(\mu)
  \;=\; \frac{R_2(\mu)}{R_1(\mu)^2}
  \;=\; Q_\ell\,,
  \qquad
  m_*^\ell(\mu) \;=\; R(\mu)^2\,m_*^\ell\,.
  \label{eq:Qell_RG_invariance}
\end{equation}
The Koide ratio is thus exactly RG-invariant at one-loop QED, and the
entire continuous flow of the lepton-extension scale collapses to a
single intensive trajectory governed by $R(\mu)$. With one-loop QED
running and threshold-matched $\alpha_{\text{em}}(\mu)$ anchored at
$\alpha_{\text{em}}^{-1}(M_Z) = 127.952$,
$m_*^\ell(\mu)$ varies monotonically from $1.255\,\text{GeV}$ at
$\mu = 1\,\text{GeV}$ to $1.213\,\text{GeV}$ at $\mu = 10\,\text{TeV}$,
a $3\%$ excursion over four decades, and crosses no Standard Model
mass within that range. The $\sim 1.9\times 10^{-3}$ shift in
$Q_\ell$ between its low-scale value and the $M_Z$ figure of
Eq.~(5.21) is therefore a subleading effect (higher-loop QED, weak
corrections, and pole-to-$\overline{\text{MS}}$ scheme conversion in
the inputs of Ref.~\cite{XingZhangZhou2008}), not a leading-order
running of $Q_\ell$ itself. Within the universal one-loop QED regime
the discrete five-scale scan of Sec.~\ref{sec:scan} reduces to repeated
samples of this single curve.

\section{Foot-Angle Interpretation}
\label{sec:foot-angle}

Following Foot's geometric reformulation~\cite{Foot1994}, let
$r=(\sqrt{m_1},\dots,\sqrt{m_N})$ and define $\theta_N$ as the angle
between $r$ and the democratic direction
$\mathbf{1}=(1,\dots,1)$. In model classes where the charged-lepton
spectrum is generated from family-space bilinears or aligned vacuum
expectation values, the mass amplitudes $r_i=\sqrt{m_i}$ can be more natural
flavor coordinates than the masses $m_i$ themselves; various family-space
mechanisms along these lines have been proposed~\cite{Koide1990FamilyHiggs,Koide2006TripletHiggs,Koide2009Yukawaon},
but no unique derivation of this viewpoint is presently established. Then
\begin{equation}
  Q_N = \frac{\|r\|^2}{(\mathbf{1}\!\cdot\! r)^2}
      = \frac{1}{N\cos^2\theta_N}\,.
  \label{eq:q_theta_relation}
\end{equation}
Equivalently,
\begin{equation}
  \theta_N = \arccos\!\left(\frac{1}{\sqrt{NQ_N}}\right) .
  \label{eq:theta_from_q}
\end{equation}
Combining Eq.~\eqref{eq:q_theta_relation} with
$N_{\mathrm{eff}}\equiv 1/Q$ gives the exact dictionary
\begin{equation}
  \cos^2\theta_N = \frac{N_{\mathrm{eff}}}{N}\,.
  \label{eq:theta_neff_dictionary}
\end{equation}
Along the one-particle extension line with base moments
$R_1=\sum_i r_i$ and $R_2=\sum_i r_i^2$, this becomes
\begin{equation}
  \theta_{N+1}(r)
  = \arccos\!\left(
      \frac{R_1+r}{\sqrt{N+1}\sqrt{R_2+r^2}}
    \right) .
  \label{eq:theta_extension_line}
\end{equation}
Using the Lorentzian form of Eq.~\eqref{eq:neff_lorentzian}, one therefore has
\begin{equation}
  \cos^2\theta_{N+1}(u)
  = \frac{N_{\mathrm{eff}}(u)}{N+1}
  = \frac{N_{\mathrm{eff},\max}}{(N+1)(1+u^2)}
  = \frac{\cos^2\theta_*}{1+u^2}\,.
  \label{eq:theta_u_relation}
\end{equation}
Equivalently,
\begin{equation}
  \theta_{N+1}(u)
  = \arccos\!\left(\frac{\cos\theta_*}{\sqrt{1+u^2}}\right) .
  \label{eq:theta_from_u}
\end{equation}
So the Foot-angle miss is quadratic near the optimum,
$\theta_{N+1}(u)=\theta_*+u^2/(2\tan\theta_*)+O(u^4)$, as expected at the
center of an even Lorentzian profile.
Thus $Q$ is equivalently an angular misalignment variable: smaller $Q$
means better alignment with the democratic direction. Propagating the
quoted mass uncertainties through this relation gives for the observed
charged leptons
$\theta_\ell = 0.785\,393\,9(9)\,\text{rad} = 44.999\,756(52)^\circ$,
extremely close to the exact Koide angle $\pi/4$~\cite{Foot1994}.
Along the constrained extension line where the original mass amplitudes are
held fixed, the theorem-selected value $r^*=R_2/R_1$ is precisely the
point of least angular deviation from the democratic direction.
Correspondingly,
\begin{equation}
  \theta_*
  = \theta_{N+1}(r^*)
  = \arccos\!\left(\frac{1}{\sqrt{(N+1)Q_{N+1,\min}}}\right)
  = \arccos\!\left(\sqrt{\frac{1+Q_0}{(N+1)Q_0}}\right) .
  \label{eq:theta_star}
\end{equation}
The theorem-selected four-body minimum
$Q_{4,\min}^{\mathrm{exp}} = 0.399\,997\,8(43)$ corresponds to
\begin{equation}
  \theta_* = 0.659\,054(7)\,\text{rad} = 37.761\,04(40)^\circ\,.
  \label{eq:theta_star_numeric}
\end{equation}
The physical charm extension of Eq.~\eqref{eq:Qext},
$Q(e,\mu,\tau,c)=0.400\,002\,5(64)$, gives
\begin{equation}
  \theta_c = 0.659\,062(10)\,\text{rad} = 37.761\,48(59)^\circ\,.
  \label{eq:theta_c_numeric}
\end{equation}
The angular miss is therefore only
\begin{equation}
  \Delta\theta_c = \theta_c-\theta_* = 0.000\,008(12)\,\text{rad}
  = 0.000\,43(71)^\circ\,.
  \label{eq:delta_theta_c}
\end{equation}
So the charm point is best viewed as a very small democratic-misalignment
perturbation away from the exact theorem minimum. The uncertainty on
$\Delta\theta_c$ is larger than its central value because it is obtained
by subtracting two very close angles with independent propagated errors;
that does not mean $r_c-r_*^\ell$ vanishes within error, only that the
angular distinction is not by itself statistically sharp.

The angle $\theta_c$ is not, however, close to any especially simple
rational multiple of $\pi$. For example, it lies
$1.761\,48(59)^\circ$ above $\pi/5$.
So what is distinctive here is not proximity to a simple standalone angle,
but proximity to the theorem-selected democratic angle associated with the
charged-lepton base itself.

\section{Look-Elsewhere Calibration}
\label{sec:scan}

To quantify where the charm continuation sits within the non-neutrino SM
one-particle extensions, we performed an exhaustive scan of all distinct
one-particle extensions of 2-body and 3-body seeds drawn from
$\{e,\mu,\tau,u,d,s,c,b,t,W,Z,H\}$. The 2-/3-body restriction is a
deliberate choice of comparison class: 1-body bases give the trivial
$Q(m)=1$ problem whose extension reduces to a single-mass match and
adds no useful comparison value, while 4-body bases of distinct
non-neutrino SM species yield only $\binom{12}{4}=495$ patterns from
which 8 extensions each can be drawn, a comparison class so small
that it is dominated by the few clusters already containing the
charged-lepton triplet and is not informative as a look-elsewhere
sample. The 2-/3-body window is therefore the smallest non-trivial
choice that gives both an interpretable rank and a sample size large
enough to populate the low-miss tail. If a trial contains
quarks, all quark masses in that trial are evaluated at the same common
$\overline{\text{MS}}$ scale chosen from
$\mu\in\{2\,\text{GeV},m_c,m_b,m_Z,m_t\}$. Purely non-quark trials are
scale-independent and are evaluated once with pole masses for
$e,\mu,\tau,W,Z,H$. The scan uses unordered bases with distinct species,
giving $660$ pair-base patterns and $1{,}980$ triple-base patterns, for
$2{,}640$ collapsed physical patterns in total. The corresponding raw
trial count is $120 + 5(2{,}640-120)=12{,}720$: the $120$ purely
non-quark patterns appear once, while the remaining $2{,}520$
quark-containing patterns are evaluated at five common scales. Here
``common-scale'' means common-scale for the quark masses within a given
trial; the charged leptons and electroweak bosons remain at their usual
reference masses (pole masses for $e,\mu,\tau,W,Z,H$) regardless of the
quark scale, so the scan inherits the same mixed-definition convention
that we labeled pragmatic for the standalone charm comparison. For the
scan itself, the quark masses are generated from
the PDG input values by one-loop QCD running with
$\alpha_s(M_Z)=0.1179$, using continuous matching of $\alpha_s$ and the
running masses across the $m_c$, $m_b$, and $m_t$ thresholds and changing
the active flavor number stepwise at those thresholds. The seed inputs
are the PDG~2024 short-distance values $m_q(m_q)$ for $q\in\{c,b,t\}$ and
$m_q(2\,\text{GeV})$ for $q\in\{u,d,s\}$, run from their reference scales
to each target $\mu$ in the grid; no pole-to-$\overline{\text{MS}}$
conversion is applied beyond what is implicit in the PDG input choice.
This is the
common-scale prescription used to generate the leaderboard tables below.
It is distinct from the separate $M_Z$ paragraph of Sec.~\ref{sec:scheme},
which is an external literature cross-check based on the running masses
quoted in Ref.~\cite{XingZhangZhou2008} rather than a row taken from the
scan pipeline itself. We ranked the trials by
the relative miss from the measured-input theorem minimum,
\begin{equation}
  \Delta_Q \equiv 10^6\,\frac{|Q_\text{ext}-Q_\text{min}|}{Q_\text{min}}\,.
\end{equation}
This gives a controlled benchmark ranking that avoids mixed-scale
quark assignments within a single trial. The resulting raw and collapsed
fractions are scan frequencies within this chosen benchmark set, not
formal p-values derived from a probabilistic ensemble. Two further
caveats apply to how the rank statistics should be read. First, the
common-scale grid $\mu\in\{2\,\text{GeV},m_c,m_b,m_Z,m_t\}$ contains
$\mu=m_c$, which is precisely the scale at which the
charged-lepton-plus-charm realization is closest to its theorem
minimum; the grid therefore overlaps a scale that is favorable for
the row of interest. Second, the collapse rule keeps the best-of-five
common-scale realizations per physical pattern, which is itself a
selection on the same ranking variable $\Delta_Q$ and favors patterns
whose theorem-selected scale happens to align with one of the grid
scales. The leaderboard rank of any specific row, including the
charged-lepton-plus-charm one, can therefore shift by several
positions under nearby grid redefinitions; the rank statements below
should be read as conditional on the specific grid and collapse rule
adopted here, not as grid-independent quantities. Because the
benchmark class is defined by the non-neutrino SM spectrum, one-particle
extensions of 2-body and 3-body seeds, and the small conventional common
scale set above, these fractions should be read only conditionally on
that chosen ensemble. They are therefore benchmark ranks within a partly
structured and partly a posteriori comparison class, not global measures of
how surprising the charged-lepton-plus-charm continuation would be in a
broader flavor-model space.

\begin{table}[t]
  \centering
  \caption{Benchmark-scan rank of the charged-lepton-plus-charm result in
  the common-scale scan. Here $\Delta_Q$ is the relative ppm miss from the
  theorem minimum. ``Raw'' counts each common-scale realization as a
  separate trial, while ``collapsed'' keeps only the sharpest
  common-scale realization for each physical species pattern $(\text{base},X)$.}
  \label{tab:specialness}
  \begin{tabular}{lccc}
    \toprule
    Scan scope & Trials & Rank of $e,\mu,\tau\to c$ & Fraction at or better \\
    \midrule
    Raw all trials & 12{,}720 & 33 & 0.259\% \\
    Raw 3-body bases only & 9{,}660 & 30 & 0.311\% \\
    Collapsed all patterns & 2{,}640 & 24 & 0.909\% \\
    Collapsed 3-body patterns only & 1{,}980 & 21 & 1.061\% \\
    \bottomrule
  \end{tabular}
\end{table}

\begin{figure}[t]
  \centering
  \includegraphics[width=\linewidth]{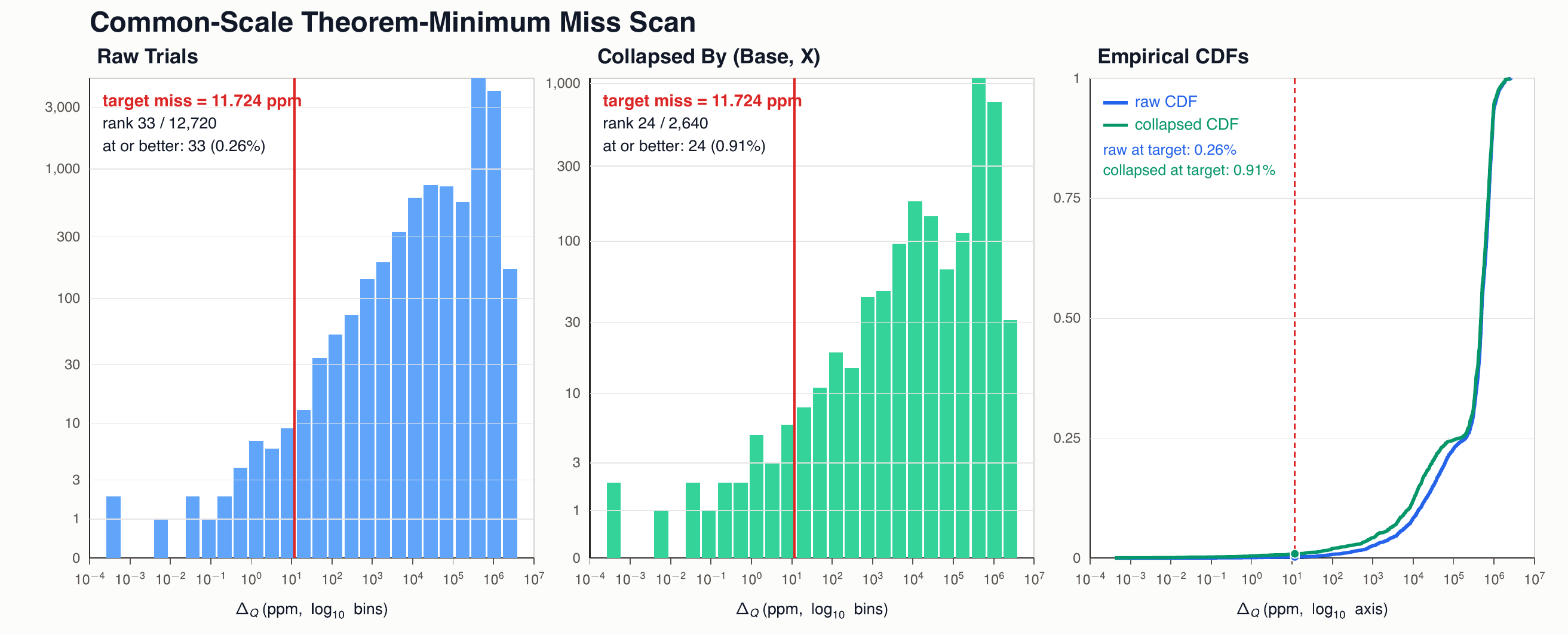}
  \caption{Distribution of the relative miss
  $\Delta_Q = 10^6 |Q_\text{ext}-Q_\text{min}|/Q_\text{min}$ in the
  common-scale scan. The left and middle panels show the raw and collapsed
  trial sets; the right panel compares the raw and collapsed empirical
  CDFs. The red reference line indicates the charged-lepton-plus-charm
  result. The raw panel counts each common-scale realization separately,
  while the collapsed panel keeps only the best common-scale realization
  for each physical species pattern $(\text{base},X)$. The figure makes
  clear that the charm continuation lies near the low-miss end of the
  distribution, but within a broader population of electroweak/heavy-sector
  near-minima.}
  \label{fig:specialness_hist}
\end{figure}

For the charged-lepton seed, the target row is
$Q(e,\mu,\tau,c(m_c))=0.400\,002\,5$, which lies
$11.7\,\text{ppm}$ above the measured-input minimum
$Q_{4,\min}^{\mathrm{exp}}=0.399\,997\,8$. Table~\ref{tab:specialness}
shows that this places the charm continuation near the low-miss end of the
distribution, but not as a unique outlier: in the common-scale scan it is
roughly a top-0.3\% result in the raw trial set and about a top-1\% result after
collapsing repeated scale realizations to one physical pattern.
Figure~\ref{fig:specialness_hist} displays the same picture graphically:
the raw and collapsed $\Delta_Q$ distributions from the common-scale
scan both place the charged-lepton-plus-charm entry (red reference line)
in the low-miss tail, but within a broader population of similarly close
near-minima rather than as a singular outlier.
The best-ranked collapsed near-minima are dominated by electroweak/heavy-quark
clusters such as $(\mu,b,H)\to Z$ and $(\tau,b,H)\to W$. The charm case
therefore remains uncommon, but its status is empirical
rather than statistically singular. The first 30 entries of the collapsed
common-scale leaderboard are listed explicitly in
Appendix~\ref{app:specialness-table}, specifically in
Table~\ref{tab:specialness-appendix}. Within that full collapsed
leaderboard, the charged-lepton-plus-charm entry appears as row 24,
marked with a star. It is worth reading that row against the rows
immediately above it, since the better-ranked entries are mostly dense
electroweak or heavy-quark clusters rather than comparably direct
continuations of the charged-lepton Koide pattern.
This electroweak/heavy-sector dominance is itself informative: once a
base tuple already lies in the dense $W$/$Z$/$H$/top region, nearby
minima are often driven by local spectral clustering rather than by a
continuation of an already distinguished flavor relation. Restricting the
collapsed leaderboard to fermions only leaves 756 patterns, within which
$e,\mu,\tau\to c$ ranks 6th (top $0.79\%$); the leading entries of that
restricted ranking are listed in
Appendix~\ref{app:specialness-table}, specifically in
Table~\ref{tab:specialness-fermion-only}. In that narrower comparison
class the charm case remains one of the best-ranked fermionic
continuations.
Table~\ref{tab:specialness-appendix} also shows the low-scale
$(e,\mu)\to s$ echo at row 7. That case is numerically sharp, but it is
less interesting physically because the base pair $Q(e,\mu)$ is not tied
to a comparably simple distinguished target. Its apparent closeness is
driven by the specific low-scale choice for the strange-quark mass, and,
as in the charm comparison, the level of agreement changes under
reasonable changes of scheme or scale.
Numerically, the $(e,\mu)$ base has
$Q(e,\mu)=0.8784119487(12)$, the theorem selects
$m_*^{e\mu}=93.2604655(22)\,\text{MeV}$ and
$Q_{3,\min}^{e\mu}=0.46763541369(33)$, while the low-scale strange input
$m_s(2\,\text{GeV})=93.4(9)\,\text{MeV}$ gives
$Q(e,\mu,s)=0.46763548(84)$, a central miss of $0.14\,\text{ppm}$ with an
uncertainty of $180\,\text{ppm}$. At current strange-mass precision this
is therefore not a statistically informative sub-ppm coincidence; it is
best regarded as a low-scale control example whose central value happens
to sit very close to the theorem minimum.
By contrast, the charm extension $Q(e,\mu,\tau,c)=0.400\,002\,5(64)$
has a propagated relative uncertainty of about $16\,\text{ppm}$, of the
same order as its $11.7\,\text{ppm}$ central miss, so the
charm coincidence is statistically informative at roughly the
$\lesssim 1\sigma$ level rather than buried inside an order-of-magnitude
larger error band as for $(e,\mu)\to s$.
This is the key qualitative distinction in the scan: several patterns are
close to their theorem minima, but $Q(e,\mu,\tau,c)$ is the clearest
case in which the extension is numerically close \emph{and} the unextended base
$Q(e,\mu,\tau)$ is already very close to an interesting simple
fraction, namely the Koide value $2/3$.

The same theorem can also serve as a mass estimator once a base tuple is
regarded as physically meaningful but the extension mass is not yet known.
That is the logic behind the charm application, and it is why the theorem
may still be useful outside the charged-lepton case: if one had a credible
Koide-like target for a neutrino sector or another hidden multiplet, the
closed-form minimizer $m^*=(R_2/R_1)^2$ would immediately turn the known
members into a predicted extension scale. In the present paper we do not
push that logic for neutrinos, because any such target assignment is much
more speculative than the charged-lepton Koide relation itself.

\section{Conclusions}
\label{sec:conclusions}

We have proved and applied the $Q$-minimum theorem: for any
$N$-particle group with Koide ratio $Q_0$, the minimum of $Q$ over
all extensions by a single additional mass is
$Q_\text{min} = Q_0/(1+Q_0)$, achieved at the
\emph{optimal mass} $m^* = (R_2/R_1)^2$.
In the associated effective-participant language,
$N_{\mathrm{eff}}\equiv 1/Q$, the one-particle theorem reads
$N_{\mathrm{eff},\max}=N_{\mathrm{eff},0}+1$, while the equal-$k$
multiplet extension theorem reads
$N_{\mathrm{eff},\max}^{(k)}=N_{\mathrm{eff},0}+k$.

The theorem itself is purely kinematic. It does not explain why any
physical spectrum should lie near the minimizing extension; rather, it
defines a unique benchmark against which candidate continuations can be
measured. The exact Lorentzian-coordinate reformulation developed in
Sec.~\ref{sec:neff_lorentzian} should be read in the same spirit: it is an
exact reparameterization of the kinematic extension curve, not a literal
dynamical Breit--Wigner resonance.
No global statistical significance is claimed for the charm proximity
beyond its rank within the explicit benchmark ensemble defined in
Sec.~\ref{sec:scan}.

Applied to the charged leptons:
\begin{itemize}
  \item The idealized value $Q_\ell^{\mathrm{K}} = 2/3$ implies the exact
        four-body benchmark value $Q_{4,\min}^{\mathrm{K}} = 2/5$.
        In the same idealized language this is
        $N_{\mathrm{eff},0}^{\mathrm{K}}=3/2 \to N_{\mathrm{eff},\max}^{\mathrm{K}}=5/2$.
  \item The measured lepton input gives the minimizing mass
        $m_*^\ell = 1.255\,34(16)\,\text{GeV}$ and the corresponding minimum
        $Q_{4,\min}^{\mathrm{exp}} = 0.399\,997\,8(43)$.
        Equivalently, the measured seed and theorem-selected peak values are
        $N_{\mathrm{eff},0}^{\mathrm{exp}}=1.500014(27)$ and
        $N_{\mathrm{eff},\max}^{\mathrm{exp}}=2.500014(27)$.
  \item The conventional own-scale short-distance charm mass
          $m_c(m_c) = 1.273(9)\,\text{GeV}$
          is $1.4\%$ above $m_*^\ell$, yielding
          $Q(e,\mu,\tau,c) = 0.400\,002\,5(64)$, which lies
          $11.7\,\text{ppm}$ above $Q_{4,\min}^{\mathrm{exp}}$ and,
           in the idealized Koide limit, $6.2\,\text{ppm}$ above
          $Q_{4,\min}^{\mathrm{K}} = 2/5$.
          In the exact Lorentzian coordinate it corresponds to
          $u_c=+0.0034(17)$, i.e. essentially the peak center.
  \item Among the conventional charm-mass choices considered here,
        $m_c(m_c)$ gives the closest realization within the low-energy
        comparison class used here. This comparison is intentionally mixed-definition: lepton pole
        masses are kept because that is the empirical charged-lepton Koide
        input, while charm is represented by the short-distance quantity
        $m_c(m_c)$ because a charm pole mass is not comparably clean.
        This gives a practical phenomenological convention, not a
        model-derived common scheme; when all four masses are instead moved
        to the common scale $M_Z$, one finds
         $Q(e,\mu,\tau,c;M_Z)=0.41098(39)$, about $2.63(98)\times 10^4$ ppm
         above the theorem-selected minimum, so the low-energy charm
         proximity does not survive as a common-scale statement.
  \item In Foot's angular language, the near-$45^\circ$ lepton geometry is
        the classic observation of Ref.~\cite{Foot1994}, while the new point
        here is that the theorem selects the least-misaligned continuation
        angle $\theta_* = 0.659\,054(7)\,\text{rad} = 37.761\,04(40)^\circ$
        and that the physical
        charm extension lands at
        $\theta_c = 0.659\,062(10)\,\text{rad} = 37.761\,48(59)^\circ$,
        i.e. only
        $\Delta\theta_c = 0.000\,008(12)\,\text{rad} = 0.000\,43(71)^\circ$
        away from that
         theorem-selected democratic direction.
  \item In a common-scale look-elsewhere scan over all non-neutrino SM
        2-body and 3-body bases, the charged-lepton-plus-charm pattern ranks
        $33/12{,}720$ in the raw trial set and $24/2{,}640$ after collapsing
        repeated common-scale realizations, so it lies near the low-miss end of
        the distribution but is not unique. Restricting to fermion-only
        collapsed patterns moves it to rank $6/756$. The full collapsed scan
        also contains a numerically sharper low-scale $(e,\mu)\to s$ echo,
        but it is less structurally compelling because it starts from the
        nondistinguished two-body value $Q(e,\mu)$ and its closeness is more
        convention-dependent.
\end{itemize}

The $Q$-minimum theorem itself is exact. Its phenomenological application to the
charm quark remains empirical: the theorem identifies a minimizing mass
scale, and the physical charm mass lies close to it in one
practical mixed-definition convention. The exhaustive 2-body and 3-body
scan shows that this proximity is uncommon but not unique once the
nearby Standard Model alternatives are enumerated in a controlled
quark-common-scale benchmark setup. That is the main claim of the paper: not that the
theorem explains flavor, but that it defines a benchmark and that
the charged-lepton-plus-charm continuation lands high in the
resulting Standard Model ranking, especially within the narrower
fermion-only comparison class.
Those rank statements are conditional on the benchmark class chosen here
and should not be reinterpreted as global p-values or prior probabilities
for flavor structure.

Within the broader Koide literature, the principal new results of this
work are the exact one-particle minimization theorem itself, its
closed-form extension rule $Q_\text{min}=Q_0/(1+Q_0)$ with minimizer
$m^*=(R_2/R_1)^2$, its equal-multiplet extension
$Q_{\min}^{(k)}=Q_0/(1+kQ_0)$ with all added amplitudes fixed at the
same $r^*=R_2/R_1$ (in particular, the classic Koide value
$Q_0=2/3$ gives the exact one-particle minimum $Q_{\min}=2/5$), its
effective-participant reformulation
$N_{\mathrm{eff},\max}=N_{\mathrm{eff},0}+1$ and
$N_{\mathrm{eff},\max}^{(k)}=N_{\mathrm{eff},0}+k$, the exact Lorentzian
coordinate form of the one-particle extension curve, and the corresponding Standard Model calibration,
including the explicit charged-lepton-plus-charm benchmark
$Q(e,\mu,\tau,c)=0.400\,002\,5(64)$ together with its scan-based
context.
Those results do not supply a dynamical flavor model, but they do define
a reproducible benchmark that may be useful in future work on
Koide-inspired extensions or hidden-sector multiplets.
The limited outlook is therefore simply this: whenever a base multiplet is
independently motivated, the theorem supplies a unique continuation scale
against which any candidate additional state can be compared. Whether that
benchmark has deeper dynamical meaning is a separate question not pursued
here.

\begin{acknowledgments}
The author thanks Christian Schmidt and Samir Varma for valuable
correspondence and the Particle Data Group for maintaining the data
compilations used throughout this work.
\end{acknowledgments}

\appendix

\section{Explicit \texorpdfstring{$r$}{r}-space computation of the measured-input minimum}
\label{app:explicit}

Using PDG~2024 lepton masses (in MeV),
\begin{align}
  m_e &= 0.51099895000(15)\,, \\
  m_\mu &= 105.6583755(23)\,, \\
  m_\tau &= 1776.86(21)\,.
\end{align}
The corresponding mass amplitudes are
\begin{align}
  r_e &= \sqrt{m_e} = 0.71484190560(11)\,\text{MeV}^{1/2}\,, \\
  r_\mu &= \sqrt{m_\mu} = 10.27902600(11)\,\text{MeV}^{1/2}\,, \\
  r_\tau &= \sqrt{m_\tau} = 42.1528(25)\,\text{MeV}^{1/2}\,.
\end{align}
Hence
\begin{align}
  R_1 &= r_e + r_\mu + r_\tau
       = 53.1467(25)\,\text{MeV}^{1/2}\,, \\
  R_2 &= r_e^2 + r_\mu^2 + r_\tau^2
       = m_e + m_\mu + m_\tau
       = 1883.029(210)\,\text{MeV}\,, \\
  Q_\ell^{\mathrm{exp}} &= \frac{R_2}{R_1^2}
       = \frac{1883.029}{2824.572}
       = 0.666\,661(12)\,, \\
  r^* &= \frac{R_2}{R_1}
       = \frac{1883.029}{53.1467}
       = 35.4308(23)\,\text{MeV}^{1/2}\,, \\
  m_*^\ell &= (r^*)^2 = 1255.34(16)\,\text{MeV}
       = 1.25534(16)\,\text{GeV}\,.
\end{align}
The relative uncertainty on $R_2$ is $\delta R_2/R_2 \approx 1.1\times 10^{-4}$,
inherited essentially entirely from $\delta m_\tau/m_\tau\approx 1.2\times 10^{-4}$,
because $m_\tau$ dominates the sum.
The smaller relative uncertainty on $Q_\ell^{\mathrm{exp}}$
($\approx 1.8\times 10^{-5}$) and on $Q_{4,\min}^{\mathrm{exp}}$
($\approx 1.1\times 10^{-5}$) reflects partial cancellation:
$Q=R_2/R_1^2$ depends on $m_\tau$ both through $R_2\propto m_\tau$ and
through $R_1\supset \sqrt{m_\tau}$, so
$\partial\ln Q/\partial\ln m_\tau = m_\tau/R_2 - r_\tau/R_1
\approx 0.943 - 0.793 = 0.150$, suppressing the propagated $Q$ error by
roughly an order of magnitude relative to the input. The further
suppression on $Q_{4,\min}=Q_\ell/(1+Q_\ell)$ is the kinematic factor
$\partial Q_{4,\min}/\partial Q_\ell = 1/(1+Q_\ell)^2 \approx 0.36$.
The theorem then gives
\begin{equation}
  Q_{4,\min}^{\mathrm{exp}} = \frac{Q_\ell^{\mathrm{exp}}}{1+Q_\ell^{\mathrm{exp}}}
  = \frac{0.666\,661}{1.666\,661}
  = 0.399\,997\,8(43)\,.
\end{equation}
Direct evaluation of the charm extension with
$m_c = 1273(9)\,\text{MeV}$ and
$r_c = \sqrt{m_c} = \sqrt{1273\,\text{MeV}} = 35.679(126)\,\text{MeV}^{1/2}$ gives
\begin{align}
  R_1' &= R_1 + r_c = 88.8258(126)\,\text{MeV}^{1/2}\,, \\
  R_2' &= R_2 + r_c^2 = 3156.0(9.0)\,\text{MeV}\,, \\
  Q_4 &= \frac{R_2'}{(R_1')^2}
      = \frac{3156.029}{7889.998}
      = 0.400\,002\,5(64)\,.
\end{align}
The corresponding measured effective-participant values and Lorentzian coordinate are
\begin{align}
  N_{\mathrm{eff},0}^{\mathrm{exp}} &= \frac{1}{Q_\ell^{\mathrm{exp}}}
    = 1.500\,014(27)\,, \\
  N_{\mathrm{eff},\max}^{\mathrm{exp}} &= \frac{1}{Q_{4,\min}^{\mathrm{exp}}}
    = 2.500\,014(27)\,, \\
  u_c &= \frac{p_c-Q_{4,\min}^{\mathrm{exp}}}{\sqrt{Q_{4,\min}^{\mathrm{exp}}(1-Q_{4,\min}^{\mathrm{exp}})}}
      = +0.0034(17)\,,
\end{align}
where $p_c=r_c/(R_1+r_c)=0.40168(85)$ is the normalized amplitude share of
the charm extender.
This is the measured-input realization of the exact identity
$N_{\mathrm{eff},\max}=N_{\mathrm{eff},0}+1$ discussed in the main text.
The uncertainty on $R_2'$ is dominated by $\delta m_c$, whereas the
uncertainty on $Q_4$ is much smaller because the charm input lies close
to the minimizing point, where the first derivative vanishes.
The residuals are
\begin{align}
  \frac{Q_4 - 2/5}{2/5} &= 6.18\times 10^{-6}
  \qquad (6.2\,\text{ppm})\,, \\
  \frac{Q_4 - Q_{4,\min}^{\mathrm{exp}}}{Q_{4,\min}^{\mathrm{exp}}}
    &= 1.17\times 10^{-5}
  \qquad (11.7\,\text{ppm})\,.
\end{align}

\section{Collapsed common-scale look-elsewhere leaderboard}
\label{app:specialness-table}

Table~\ref{tab:specialness-appendix} lists the first 30 entries of the
collapsed common-scale leaderboard discussed in
Sec.~\ref{sec:scan}. Here ``collapsed'' means that for each
physical species pattern $(\text{base},X)$ we keep only the common-scale
realization with the smallest relative ppm miss from the theorem minimum. In
the tables below, $m^*$ is the theorem-selected extension mass, $m_X$ is
the actual mass of the candidate extender $X$ at the listed scale,
$\delta_m \equiv 10^6(m_X-m^*)/m^*$ is their signed relative offset in
ppm, $u_{\mathrm{obs}}$ is the Lorentzian-coordinate value of the observed
extender defined in Eq.~\eqref{eq:u_coordinate}, and
$\Delta\theta \equiv \theta_{\mathrm{obs}}-\theta_*$ is the
corresponding Foot-angle miss, quoted in $\mu$rad.

\begin{table}[p]
  \centering
  \caption{First 30 entries of the collapsed common-scale leaderboard,
  ranked by $\Delta_Q \equiv 10^6 |Q_\text{obs}-Q_\text{min}|/Q_\text{min}$.
  The mass columns $m^*$ and $m_X$ are in GeV. The angle columns
  $\theta_*$ and $\theta_{\mathrm{obs}}$ are in rad, while the final
  column gives the Foot-angle miss $\Delta\theta$ in $\mu$rad. The signed
  mass-offset column is $\delta_m \equiv 10^6(m_X-m^*)/m^*$, and the
  $u_{\mathrm{obs}}$ column is the Lorentzian-coordinate value from
  Eq.~\eqref{eq:u_coordinate} with propagated $1\sigma$ uncertainty.
  A star marks the charged-lepton-plus-charm target row. The scale column
  lists the common quark evaluation scale; ``indep.'' denotes a pure
  non-quark pattern. Several entries in this table have
  $|u_{\mathrm{obs}}|\gtrsim 0.3$, well outside the small-deviation
  regime where the Lorentzian profile of
  Eq.~\eqref{eq:neff_lorentzian} retains its peak character; the listed
  $u_{\mathrm{obs}}$ values for those rows should be read as positions
  on the exact kinematic curve rather than as ``near-peak'' realizations.}
  \label{tab:specialness-appendix}
  \setlength{\tabcolsep}{2.5pt}
  \resizebox{\textwidth}{!}{%
  \begin{tabular}{r l l r r r r r r r r r r}
    \toprule
    Rank & Pattern & Scale & $Q_\text{min}$ & $Q_\text{obs}$ & $\Delta_Q$ & $m^*$ & $m_X$ & $\delta_m$ & $u_{\mathrm{obs}}$ & $\theta_*$ & $\theta_{\mathrm{obs}}$ & $\Delta\theta$ \\
    \midrule
    1 & $\mu,b,H\to Z$ & $m_b$ & 0.413238900 & 0.413238900 & 0.000411 & 91.195134 & 91.187621 & -82.381 & $-0.00013(25)$ & 0.67963(16) & 0.67963(16) & 0.0003(65) \\
    2 & $\tau,b,H\to W$ & 2 GeV & 0.379081210 & 0.379081210 & 0.000487 & 80.376505 & 80.369193 & -90.969 & $-0.00218(26)$ & 0.62307(19) & 0.62307(19) & 0.0003(80) \\
    3 & $\mu,s,Z\to W$ & $m_c$ & 0.467992686 & 0.467992689 & 0.007778 & 80.397609 & 80.369193 & -353.438 & $-0.00066(8)$ & 0.75118(55) & 0.75118(55) & 0.004(45) \\
    4 & $t,W\to H$ & $m_t$ & 0.340010524 & 0.340010532 & 0.022759 & 125.120300 & 125.200015 & +637.106 & $-0.01056(42)$ & 0.1406(19) & 0.1406(19) & 0.1(19) \\
    5 & $\mu,b,t\to H$ & $m_t$ & 0.430746950 & 0.430746960 & 0.023921 & 125.278251 & 125.200015 & -624.492 & $-0.00933(59)$ & 0.70466(42) & 0.70466(42) & 0.01(50) \\
    6 & $\tau,c,H\to Z$ & $m_t$ & 0.416276923 & 0.416276958 & 0.084543 & 91.080115 & 91.187621 & +1180.344 & $+0.00789(25)$ & 0.68415(16) & 0.68415(16) & 0.052(95) \\
    7 & $e,\mu\to s$ & 2 GeV & 0.467635414 & 0.467635479 & 0.139122 & 0.093260 & 0.093400 & +1496.181 & $+0.00037(249)$ & 0.56557910660(56) & 0.5655792(94) & 0.1(94) \\
    8 & $\mu,t,W\to H$ & $m_Z$ & 0.334326147 & 0.334326255 & 0.324022 & 125.502579 & 125.200015 & -2410.815 & $-0.00385(42)$ & 0.52617(54) & 0.52617(54) & 0.3(18) \\
    9 & $e,t,W\to H$ & $m_t$ & 0.339544132 & 0.339544272 & 0.413088 & 124.860701 & 125.200015 & +2717.541 & $-0.01008(42)$ & 0.53930(46) & 0.53930(45) & 0.3(19) \\
    10 & $b,W,H\to Z$ & $m_t$ & 0.303543312 & 0.303543491 & 0.590759 & 90.883361 & 91.187621 & +3347.813 & $+0.00479(14)$ & 0.43344(18) & 0.43344(18) & 0.64(25) \\
    11 & $s,c,b\to \tau$ & $m_t$ & 0.320350290 & 0.320350572 & 0.881519 & 1.769727 & 1.776860 & +4030.952 & $-0.08043(43)$ & 0.4877(30) & 0.4877(30) & 0.8(21) \\
    12 & $u,t,W\to H$ & $m_t$ & 0.339249104 & 0.339249413 & 0.909015 & 124.696947 & 125.200015 & +4034.329 & $-0.00957(44)$ & 0.53857(50) & 0.53857(50) & 0.8(29) \\
    13 & $s,b,H\to Z$ & $m_b$ & 0.414487581 & 0.414488224 & 1.552190 & 91.650117 & 91.187621 & -5046.319 & $-0.00079(26)$ & 0.68212(50) & 0.68212(50) & 0.96(63) \\
    14 & $\mu,t,Z\to H$ & $m_t$ & 0.331521875 & 0.331522390 & 1.553503 & 125.865037 & 125.200015 & -5283.608 & $-0.01141(40)$ & 0.51885(45) & 0.51886(45) & 1.4(37) \\
    15 & $d,t,W\to H$ & $m_t$ & 0.338892548 & 0.338893146 & 1.764836 & 124.499519 & 125.200015 & +5626.496 & $-0.00911(42)$ & 0.53769(47) & 0.53769(47) & 1.5(40) \\
    16 & $c,t,Z\to H$ & $m_Z$ & 0.322297431 & 0.322298345 & 2.837812 & 124.301269 & 125.200015 & +7230.387 & $+0.00261(41)$ & 0.49340(61) & 0.49340(60) & 2.6(55) \\
    17 & $\mu,c,b\to \tau$ & $m_Z$ & 0.311380167 & 0.311381062 & 2.875405 & 1.789931 & 1.776860 & -7302.211 & $-0.07882(41)$ & 0.46004(68) & 0.46004(69) & 2.9(28) \\
    18 & $s,t,W\to H$ & $m_Z$ & 0.335827143 & 0.335828648 & 4.484159 & 126.328634 & 125.200015 & -8933.987 & $-0.00424(42)$ & 0.53001(66) & 0.53001(66) & 3.8(66) \\
    19 & $e,\tau,s\to c$ & $m_c$ & 0.401297789 & 0.401300569 & 6.928563 & 1.259408 & 1.273000 & +10792.368 & $+0.00092(53)$ & 0.6611(34) & 0.6611(34) & 4.5(74) \\
    20 & $s,b,t\to H$ & $m_t$ & 0.433408931 & 0.433412261 & 7.683021 & 126.609018 & 125.200015 & -11128.773 & $-0.00992(59)$ & 0.70827(53) & 0.70827(53) & 4.5(89) \\
    21 & $s,t,Z\to H$ & $m_t$ & 0.333094273 & 0.333097021 & 8.249893 & 126.736889 & 125.200015 & -12126.487 & $-0.01178(40)$ & 0.52298(58) & 0.52298(59) & 7.2(85) \\
    22 & $\mu,\tau\to c$ & $m_c$ & 0.406449527 & 0.406453397 & 9.520976 & 1.289105 & 1.273000 & -12493.044 & $-0.00309(49)$ & 0.4380058(60) & 0.438016(13) & 10(11) \\
    23 & $c,b,H\to W$ & $m_c$ & 0.382528794 & 0.382533207 & 11.534316 & 81.501266 & 80.369193 & -13890.251 & $-0.00719(26)$ & 0.62933(28) & 0.62934(28) & 7.9(13) \\
    24$^*$ & $e,\mu,\tau\to c$ & $m_c$ & 0.399997781 & 0.400002471 & 11.723747 & 1.255341 & 1.273000 & +14066.729 & $+0.0034(17)$ & 0.6590545(40) & 0.6590620(86) & 7.6(77) \\
    25 & $\tau,s\to c$ & $m_c$ & 0.407781159 & 0.407787380 & 15.256796 & 1.293413 & 1.273000 & -15782.573 & $-0.00567(54)$ & 0.4415(56) & 0.4415(56) & 16(18) \\
    26 & $\mu,u\to s$ & $m_b$ & 0.441384763 & 0.441392989 & 18.635453 & 0.085012 & 0.083546 & -17245.942 & $+0.02575(583)$ & 0.5176(10) & 0.5176(10) & 16(120) \\
    27 & $\tau,c,t\to H$ & $m_Z$ & 0.425827919 & 0.425836775 & 20.797273 & 127.532475 & 125.200015 & -18289.142 & $-0.00298(60)$ & 0.69786(46) & 0.69787(48) & 12(15) \\
    28 & $d,b,H\to Z$ & $m_c$ & 0.416316467 & 0.416325516 & 21.734841 & 92.930262 & 91.187621 & -18752.132 & $-0.00917(25)$ & 0.68420(21) & 0.68422(21) & 13(15) \\
    29 & $\tau,W,H\to Z$ & indep. & 0.309913954 & 0.309921197 & 23.369713 & 93.118011 & 91.187621 & -20730.574 & $-0.00490(14)$ & 0.45525(53) & 0.45527(54) & 24(14) \\
    30 & $\mu,\tau,d\to c$ & 2 GeV & 0.387706554 & 0.387716195 & 24.866871 & 1.194975 & 1.170735 & -20284.763 & $+0.01546(52)$ & 0.63848(100) & 0.63849(100) & 17(12) \\
    \bottomrule
  \end{tabular}
  }
\end{table}

Table~\ref{tab:specialness-fermion-only} shows the top entries of the
collapsed leaderboard after removing $W$, $Z$, and $H$ and restricting to
the fermionic species $\{e,\mu,\tau,u,d,s,c,b,t\}$. In that reduced
comparison class the charged-lepton-plus-charm target moves from row 24
in the full table to row 6.

\begin{table}[t]
  \centering
  \caption{Top 10 entries of the fermion-only collapsed leaderboard
  (756 patterns total), obtained by restricting the scan to
  $\{e,\mu,\tau,u,d,s,c,b,t\}$. The mass columns $m^*$ and $m_X$ are
  in GeV. The angle columns $\theta_*$ and $\theta_{\mathrm{obs}}$ are
  in rad, while the final column gives the Foot-angle miss
  $\Delta\theta$ in $\mu$rad; all three angle columns use $1\sigma$
  uncertainties in APS-style parenthetical notation. The signed
  mass-offset column is $\delta_m \equiv 10^6(m_X-m^*)/m^*$, and the
  $u_{\mathrm{obs}}$ column is the Lorentzian-coordinate value from
  Eq.~\eqref{eq:u_coordinate}. For the
  charged-lepton-plus-charm target row,
  $\Delta\theta = 7.6(77)\,\mu\text{rad}$. A star marks the
  charged-lepton-plus-charm target row.}
  \label{tab:specialness-fermion-only}
  \setlength{\tabcolsep}{2.5pt}
  \resizebox{\textwidth}{!}{%
  \begin{tabular}{r l l r r r r r r r r r r}
    \toprule
    Rank & Pattern & Scale & $Q_\text{min}$ & $Q_\text{obs}$ & $\Delta_Q$ & $m^*$ & $m_X$ & $\delta_m$ & $u_{\mathrm{obs}}$ & $\theta_*$ & $\theta_{\mathrm{obs}}$ & $\Delta\theta$ \\
    \midrule
    1 & $e,\mu\to s$ & 2 GeV & 0.467635414 & 0.467635479 & 0.139122 & 0.093260 & 0.093400 & +1496.181 & $+0.00037(249)$ & 0.56557910660(56) & 0.5655792(94) & 0.1(94) \\
    2 & $s,c,b\to \tau$ & $m_t$ & 0.320350290 & 0.320350572 & 0.881519 & 1.769727 & 1.776860 & +4030.952 & $-0.08043(43)$ & 0.4877(30) & 0.4877(30) & 0.8(21) \\
    3 & $\mu,c,b\to \tau$ & $m_Z$ & 0.311380167 & 0.311381062 & 2.875405 & 1.789931 & 1.776860 & -7302.211 & $-0.07882(41)$ & 0.46004(68) & 0.46004(69) & 2.9(28) \\
    4 & $e,\tau,s\to c$ & $m_c$ & 0.401297789 & 0.401300569 & 6.928563 & 1.259408 & 1.273000 & +10792.368 & $+0.00092(53)$ & 0.6611(34) & 0.6611(34) & 4.5(74) \\
    5 & $\mu,\tau\to c$ & $m_c$ & 0.406449527 & 0.406453397 & 9.520976 & 1.289105 & 1.273000 & -12493.044 & $-0.00309(49)$ & 0.4380058(60) & 0.438016(13) & 10(11) \\
    6$^*$ & $e,\mu,\tau\to c$ & $m_c$ & 0.399997781 & 0.400002471 & 11.723747 & 1.255341 & 1.273000 & +14066.729 & $+0.0034(17)$ & 0.6590545(40) & 0.6590620(86) & 7.6(77) \\
    7 & $\tau,s\to c$ & $m_c$ & 0.407781159 & 0.407787380 & 15.256796 & 1.293413 & 1.273000 & -15782.573 & $-0.00567(54)$ & 0.4415(56) & 0.4415(56) & 16(18) \\
    8 & $\mu,u\to s$ & $m_b$ & 0.441384763 & 0.441392989 & 18.635453 & 0.085012 & 0.083546 & -17245.942 & $+0.02575(583)$ & 0.518(10) & 0.518(10) & 16(120) \\
    9 & $\mu,\tau,d\to c$ & 2 GeV & 0.387706554 & 0.387716195 & 24.866871 & 1.194975 & 1.170735 & -20284.763 & $+0.01546(52)$ & 0.63848(100) & 0.63849(100) & 17(12) \\
    10 & $\mu,\tau,s\to c$ & $m_b$ & 0.341496792 & 0.341510567 & 40.338017 & 1.019592 & 1.047213 & +27090.548 & $+0.05513(51)$ & 0.5441(37) & 0.5441(37) & 33(21) \\
    \bottomrule
  \end{tabular}
  }
\end{table}

\clearpage


\begin{thebibliography}{21}%
\makeatletter
\providecommand \@ifxundefined [1]{%
 \@ifx{#1\undefined}
}%
\providecommand \@ifnum [1]{%
 \ifnum #1\expandafter \@firstoftwo
 \else \expandafter \@secondoftwo
 \fi
}%
\providecommand \@ifx [1]{%
 \ifx #1\expandafter \@firstoftwo
 \else \expandafter \@secondoftwo
 \fi
}%
\providecommand \natexlab [1]{#1}%
\providecommand \enquote  [1]{``#1''}%
\providecommand \bibnamefont  [1]{#1}%
\providecommand \bibfnamefont [1]{#1}%
\providecommand \citenamefont [1]{#1}%
\providecommand \href@noop [0]{\@secondoftwo}%
\providecommand \href [0]{\begingroup \@sanitize@url \@href}%
\providecommand \@href[1]{\@@startlink{#1}\@@href}%
\providecommand \@@href[1]{\endgroup#1\@@endlink}%
\providecommand \@sanitize@url [0]{\catcode `\\12\catcode `\$12\catcode
  `\&12\catcode `\#12\catcode `\^12\catcode `\_12\catcode `\%12\relax}%
\providecommand \@@startlink[1]{}%
\providecommand \@@endlink[0]{}%
\providecommand \url  [0]{\begingroup\@sanitize@url \@url }%
\providecommand \@url [1]{\endgroup\@href {#1}{\urlprefix }}%
\providecommand \urlprefix  [0]{URL }%
\providecommand \Eprint [0]{\href }%
\providecommand \doibase [0]{https://doi.org/}%
\providecommand \selectlanguage [0]{\@gobble}%
\providecommand \bibinfo  [0]{\@secondoftwo}%
\providecommand \bibfield  [0]{\@secondoftwo}%
\providecommand \translation [1]{[#1]}%
\providecommand \BibitemOpen [0]{}%
\providecommand \bibitemStop [0]{}%
\providecommand \bibitemNoStop [0]{.\EOS\space}%
\providecommand \EOS [0]{\spacefactor3000\relax}%
\providecommand \BibitemShut  [1]{\csname bibitem#1\endcsname}%
\let\auto@bib@innerbib\@empty
\bibitem [{\citenamefont {Koide}(1982)}]{Koide1982}%
  \BibitemOpen
  \bibfield  {author} {\bibinfo {author} {\bibfnamefont {Y.}~\bibnamefont
  {Koide}},\ }\bibfield  {title} {\bibinfo {title} {Fermion-boson two-body
  model of quarks and leptons and {C}abibbo mixing},\ }\href
  {https://doi.org/10.1007/BF02817096} {\bibfield  {journal} {\bibinfo
  {journal} {Lett. Nuovo Cim.}\ }\textbf {\bibinfo {volume} {34}},\ \bibinfo
  {pages} {201} (\bibinfo {year} {1982})}\BibitemShut {NoStop}%
\bibitem [{\citenamefont {Koide}(1983)}]{Koide1983}%
  \BibitemOpen
  \bibfield  {author} {\bibinfo {author} {\bibfnamefont {Y.}~\bibnamefont
  {Koide}},\ }\bibfield  {title} {\bibinfo {title} {New view of quark-lepton
  mass hierarchy},\ }\href {https://doi.org/10.1103/PhysRevD.28.252} {\bibfield
   {journal} {\bibinfo  {journal} {Phys. Rev. D}\ }\textbf {\bibinfo {volume}
  {28}},\ \bibinfo {pages} {252} (\bibinfo {year} {1983})}\BibitemShut
  {NoStop}%
\bibitem [{\citenamefont {Navas}\ \emph {et~al.}(2024)\citenamefont {Navas}
  \emph {et~al.}}]{PDG2024}%
  \BibitemOpen
  \bibfield  {author} {\bibinfo {author} {\bibfnamefont {S.}~\bibnamefont
  {Navas}} \emph {et~al.} (\bibinfo {collaboration} {Particle Data Group}),\
  }\bibfield  {title} {\bibinfo {title} {Review of particle physics},\ }\href
  {https://doi.org/10.1103/PhysRevD.110.030001} {\bibfield  {journal} {\bibinfo
   {journal} {Phys. Rev. D}\ }\textbf {\bibinfo {volume} {110}},\ \bibinfo
  {pages} {030001} (\bibinfo {year} {2024})}\BibitemShut {NoStop}%
\bibitem [{\citenamefont {Foot}(1994)}]{Foot1994}%
  \BibitemOpen
  \bibfield  {author} {\bibinfo {author} {\bibfnamefont {R.}~\bibnamefont
  {Foot}},\ }\bibfield  {title} {\bibinfo {title} {A note on {Koide}'s lepton
  mass relation},\ }\href {https://doi.org/10.1142/S0217732394000186}
  {\bibfield  {journal} {\bibinfo  {journal} {Mod. Phys. Lett. A}\ }\textbf
  {\bibinfo {volume} {9}},\ \bibinfo {pages} {169} (\bibinfo {year} {1994})},\
  \Eprint {https://arxiv.org/abs/hep-ph/9402242} {arXiv:hep-ph/9402242}
  \BibitemShut {NoStop}%
\bibitem [{\citenamefont {Sumino}(2009{\natexlab{a}})}]{Sumino2009a}%
  \BibitemOpen
  \bibfield  {author} {\bibinfo {author} {\bibfnamefont {Y.}~\bibnamefont
  {Sumino}},\ }\bibfield  {title} {\bibinfo {title} {Family gauge symmetry and
  {Koide}'s mass formula},\ }\href
  {https://doi.org/10.1016/j.physletb.2008.12.048} {\bibfield  {journal}
  {\bibinfo  {journal} {Phys. Lett. B}\ }\textbf {\bibinfo {volume} {671}},\
  \bibinfo {pages} {477} (\bibinfo {year} {2009}{\natexlab{a}})},\ \Eprint
  {https://arxiv.org/abs/0812.2090} {arXiv:0812.2090} \BibitemShut {NoStop}%
\bibitem [{\citenamefont {Sumino}(2009{\natexlab{b}})}]{Sumino2009b}%
  \BibitemOpen
  \bibfield  {author} {\bibinfo {author} {\bibfnamefont {Y.}~\bibnamefont
  {Sumino}},\ }\bibfield  {title} {\bibinfo {title} {Family gauge symmetry as
  an origin of {Koide}'s mass formula and charged lepton spectrum},\ }\href
  {https://doi.org/10.1088/1126-6708/2009/05/075} {\bibfield  {journal}
  {\bibinfo  {journal} {J. High Energy Phys.}\ }\textbf {\bibinfo {volume}
  {2009}}\bibfield  {number} {\bibinfo  {number} { (05)},\ \bibinfo {pages}
  {075}},\ }\Eprint {https://arxiv.org/abs/0812.2103} {arXiv:0812.2103}
  \BibitemShut {NoStop}%
\bibitem [{\citenamefont {Li}\ and\ \citenamefont {Ma}(2005)}]{LiMa2005}%
  \BibitemOpen
  \bibfield  {author} {\bibinfo {author} {\bibfnamefont {N.}~\bibnamefont
  {Li}}\ and\ \bibinfo {author} {\bibfnamefont {B.-Q.}\ \bibnamefont {Ma}},\
  }\bibfield  {title} {\bibinfo {title} {Estimate of neutrino masses from
  {Koide}'s relation},\ }\href {https://doi.org/10.1016/j.physletb.2005.01.066}
  {\bibfield  {journal} {\bibinfo  {journal} {Phys. Lett. B}\ }\textbf
  {\bibinfo {volume} {609}},\ \bibinfo {pages} {309} (\bibinfo {year}
  {2005})},\ \Eprint {https://arxiv.org/abs/hep-ph/0505028}
  {arXiv:hep-ph/0505028} \BibitemShut {NoStop}%
\bibitem [{\citenamefont {G{\'e}rard}\ \emph {et~al.}(2006)\citenamefont
  {G{\'e}rard}, \citenamefont {Goffinet},\ and\ \citenamefont
  {Herquet}}]{GerardGoffinetHerquet2006}%
  \BibitemOpen
  \bibfield  {author} {\bibinfo {author} {\bibfnamefont {J.-M.}\ \bibnamefont
  {G{\'e}rard}}, \bibinfo {author} {\bibfnamefont {F.}~\bibnamefont
  {Goffinet}},\ and\ \bibinfo {author} {\bibfnamefont {M.}~\bibnamefont
  {Herquet}},\ }\bibfield  {title} {\bibinfo {title} {A new look at an old mass
  relation},\ }\href {https://doi.org/10.1016/j.physletb.2005.12.054}
  {\bibfield  {journal} {\bibinfo  {journal} {Phys. Lett. B}\ }\textbf
  {\bibinfo {volume} {633}},\ \bibinfo {pages} {563} (\bibinfo {year}
  {2006})},\ \Eprint {https://arxiv.org/abs/hep-ph/0510289}
  {arXiv:hep-ph/0510289} \BibitemShut {NoStop}%
\bibitem [{\citenamefont {Li}\ and\ \citenamefont {Ma}(2006)}]{LiMa2006}%
  \BibitemOpen
  \bibfield  {author} {\bibinfo {author} {\bibfnamefont {N.}~\bibnamefont
  {Li}}\ and\ \bibinfo {author} {\bibfnamefont {B.-Q.}\ \bibnamefont {Ma}},\
  }\bibfield  {title} {\bibinfo {title} {Energy scale independence of {Koide}'s
  relation for quark and lepton masses},\ }\href
  {https://doi.org/10.1103/PhysRevD.73.013009} {\bibfield  {journal} {\bibinfo
  {journal} {Phys. Rev. D}\ }\textbf {\bibinfo {volume} {73}},\ \bibinfo
  {pages} {013009} (\bibinfo {year} {2006})},\ \Eprint
  {https://arxiv.org/abs/hep-ph/0601031} {arXiv:hep-ph/0601031} \BibitemShut
  {NoStop}%
\bibitem [{\citenamefont {Xing}\ and\ \citenamefont
  {Zhang}(2006)}]{XingZhang2006}%
  \BibitemOpen
  \bibfield  {author} {\bibinfo {author} {\bibfnamefont {Z.-z.}\ \bibnamefont
  {Xing}}\ and\ \bibinfo {author} {\bibfnamefont {H.}~\bibnamefont {Zhang}},\
  }\bibfield  {title} {\bibinfo {title} {On the {Koide}-like relations for the
  running masses of charged leptons, neutrinos and quarks},\ }\href
  {https://doi.org/10.1016/j.physletb.2006.02.051} {\bibfield  {journal}
  {\bibinfo  {journal} {Phys. Lett. B}\ }\textbf {\bibinfo {volume} {635}},\
  \bibinfo {pages} {107} (\bibinfo {year} {2006})},\ \Eprint
  {https://arxiv.org/abs/hep-ph/0602134} {arXiv:hep-ph/0602134} \BibitemShut
  {NoStop}%
\bibitem [{\citenamefont {Rodejohann}\ and\ \citenamefont
  {Zhang}(2011)}]{Rodejohann2011}%
  \BibitemOpen
  \bibfield  {author} {\bibinfo {author} {\bibfnamefont {W.}~\bibnamefont
  {Rodejohann}}\ and\ \bibinfo {author} {\bibfnamefont {H.}~\bibnamefont
  {Zhang}},\ }\bibfield  {title} {\bibinfo {title} {Extended empirical fermion
  mass relation},\ }\href {https://doi.org/10.1016/j.physletb.2011.03.007}
  {\bibfield  {journal} {\bibinfo  {journal} {Phys. Lett. B}\ }\textbf
  {\bibinfo {volume} {698}},\ \bibinfo {pages} {152} (\bibinfo {year}
  {2011})},\ \Eprint {https://arxiv.org/abs/1101.5525} {arXiv:1101.5525}
  \BibitemShut {NoStop}%
\bibitem [{\citenamefont {Kartavtsev}(2011)}]{Kartavtsev2011}%
  \BibitemOpen
  \bibfield  {author} {\bibinfo {author} {\bibfnamefont {A.}~\bibnamefont
  {Kartavtsev}},\ }\bibfield  {title} {\bibinfo {title} {A remark on the
  {Koide} relation for quarks},\ }\href@noop {} {\bibfield  {journal} {\bibinfo
   {journal} {arXiv e-print}\ } (\bibinfo {year} {2011})},\ \bibinfo {note}
  {unpublished preprint},\ \Eprint {https://arxiv.org/abs/1111.0480}
  {arXiv:1111.0480} \BibitemShut {NoStop}%
\bibitem [{\citenamefont {Gao}\ and\ \citenamefont {Li}(2015)}]{Gao2015}%
  \BibitemOpen
  \bibfield  {author} {\bibinfo {author} {\bibfnamefont {G.-H.}\ \bibnamefont
  {Gao}}\ and\ \bibinfo {author} {\bibfnamefont {N.}~\bibnamefont {Li}},\
  }\bibfield  {title} {\bibinfo {title} {Explorations of two empirical formulae
  for fermion masses},\ }\href@noop {} {\bibfield  {journal} {\bibinfo
  {journal} {arXiv e-print}\ } (\bibinfo {year} {2015})},\ \bibinfo {note}
  {unpublished preprint},\ \Eprint {https://arxiv.org/abs/1512.06349}
  {arXiv:1512.06349} \BibitemShut {NoStop}%
\bibitem [{\citenamefont {Rivero}\ and\ \citenamefont
  {Gsponer}(2005)}]{Rivero2005}%
  \BibitemOpen
  \bibfield  {author} {\bibinfo {author} {\bibfnamefont {A.}~\bibnamefont
  {Rivero}}\ and\ \bibinfo {author} {\bibfnamefont {A.}~\bibnamefont
  {Gsponer}},\ }\bibfield  {title} {\bibinfo {title} {The strange formula of
  {Dr. Koide}},\ }\href@noop {} {\bibfield  {journal} {\bibinfo  {journal}
  {arXiv e-print}\ } (\bibinfo {year} {2005})},\ \bibinfo {note} {unpublished
  preprint},\ \Eprint {https://arxiv.org/abs/hep-ph/0505220}
  {arXiv:hep-ph/0505220} \BibitemShut {NoStop}%
\bibitem [{\citenamefont {Brannen}(2010)}]{Brannen2010}%
  \BibitemOpen
  \bibfield  {author} {\bibinfo {author} {\bibfnamefont {C.~A.}\ \bibnamefont
  {Brannen}},\ }\href@noop {} {\bibinfo {title} {Spin path integrals and
  generations}} (\bibinfo {year} {2010}),\ \Eprint
  {https://arxiv.org/abs/1006.3114} {arXiv:1006.3114 [physics.gen-ph]}
  \BibitemShut {NoStop}%
\bibitem [{\citenamefont {Varma}(2026)}]{Varma2026}%
  \BibitemOpen
  \bibfield  {author} {\bibinfo {author} {\bibfnamefont {S.}~\bibnamefont
  {Varma}},\ }\bibfield  {title} {\bibinfo {title} {Unified fermion mass
  ratios from orbifold conformal field theory},\ }\href
  {https://doi.org/10.1142/S0217732326500732} {\bibfield  {journal} {\bibinfo
  {journal} {Mod. Phys. Lett. A}\ }\textbf {\bibinfo {volume} {41}},\ \bibinfo
  {pages} {2650073} (\bibinfo {year} {2026})}\BibitemShut {NoStop}%
\bibitem [{\citenamefont {{U.S. Department of Justice}}\ and\ \citenamefont
  {{Federal Trade Commission}}(2010)}]{DOJFTC2010}%
  \BibitemOpen
  \bibfield  {author} {\bibinfo {author} {\bibnamefont {{U.S. Department of
  Justice}}}\ and\ \bibinfo {author} {\bibnamefont {{Federal Trade
  Commission}}},\ }\href {https://www.justice.gov/media/810916/dl?inline}
  {\emph {\bibinfo {title} {Horizontal Merger Guidelines}}},\ \bibinfo {type}
  {Tech. Rep.}\ (\bibinfo  {institution} {U.S. Department of Justice and
  Federal Trade Commission},\ \bibinfo {year} {2010})\BibitemShut {NoStop}%
\bibitem [{\citenamefont {Xing}\ \emph {et~al.}(2008)\citenamefont {Xing},
  \citenamefont {Zhang},\ and\ \citenamefont {Zhou}}]{XingZhangZhou2008}%
  \BibitemOpen
  \bibfield  {author} {\bibinfo {author} {\bibfnamefont {Z.-z.}\ \bibnamefont
  {Xing}}, \bibinfo {author} {\bibfnamefont {H.}~\bibnamefont {Zhang}},\ and\
  \bibinfo {author} {\bibfnamefont {S.}~\bibnamefont {Zhou}},\ }\bibfield
  {title} {\bibinfo {title} {Updated values of running quark and lepton
  masses},\ }\href {https://doi.org/10.1103/PhysRevD.77.113016} {\bibfield
  {journal} {\bibinfo  {journal} {Phys. Rev. D}\ }\textbf {\bibinfo {volume}
  {77}},\ \bibinfo {pages} {113016} (\bibinfo {year} {2008})},\ \Eprint
  {https://arxiv.org/abs/0712.1419} {arXiv:0712.1419 [hep-ph]} \BibitemShut
  {NoStop}%
\bibitem [{\citenamefont {Koide}(1990)}]{Koide1990FamilyHiggs}%
  \BibitemOpen
  \bibfield  {author} {\bibinfo {author} {\bibfnamefont {Y.}~\bibnamefont
  {Koide}},\ }\bibfield  {title} {\bibinfo {title} {Charged lepton mass sum
  rule from {U}(3) family higgs potential model},\ }\href
  {https://doi.org/10.1142/S0217732390002663} {\bibfield  {journal} {\bibinfo
  {journal} {Mod. Phys. Lett. A}\ }\textbf {\bibinfo {volume} {5}},\ \bibinfo
  {pages} {2319} (\bibinfo {year} {1990})}\BibitemShut {NoStop}%
\bibitem [{\citenamefont {Koide}(2006)}]{Koide2006TripletHiggs}%
  \BibitemOpen
  \bibfield  {author} {\bibinfo {author} {\bibfnamefont {Y.}~\bibnamefont
  {Koide}},\ }\bibfield  {title} {\bibinfo {title} {Seesaw mass matrix model of
  quarks and leptons with flavor-triplet higgs scalars},\ }\href
  {https://doi.org/10.1140/epjc/s10052-006-0009-5} {\bibfield  {journal}
  {\bibinfo  {journal} {Eur. Phys. J. C}\ }\textbf {\bibinfo {volume} {48}},\
  \bibinfo {pages} {223} (\bibinfo {year} {2006})}\BibitemShut {NoStop}%
\bibitem [{\citenamefont {Koide}(2009)}]{Koide2009Yukawaon}%
  \BibitemOpen
  \bibfield  {author} {\bibinfo {author} {\bibfnamefont {Y.}~\bibnamefont
  {Koide}},\ }\bibfield  {title} {\bibinfo {title} {Charged lepton mass
  relations in a supersymmetric yukawaon model},\ }\href
  {https://doi.org/10.1103/PhysRevD.79.033009} {\bibfield  {journal} {\bibinfo
  {journal} {Phys. Rev. D}\ }\textbf {\bibinfo {volume} {79}},\ \bibinfo
  {pages} {033009} (\bibinfo {year} {2009})},\ \Eprint
  {https://arxiv.org/abs/0811.3475} {arXiv:0811.3475 [hep-ph]} \BibitemShut
  {NoStop}%
\end{thebibliography}
\end{document}